\documentclass[reqno]{amsart}

\usepackage{graphicx}% Include figure files

\newtheorem{theorem}{Theorem}
\newtheorem{proposition}{Proposition}

\newtheorem{remark}{Remark}

\newcommand{\bee}{\begin{eqnarray}}
\newcommand{\eee}{\end{eqnarray}}
\newcommand{\be}{\begin{eqnarray*}}
\newcommand{\ee}{\end{eqnarray*}}
\newcommand{\R}{{\mathbb R}}
\newcommand{\N}{{\mathbb N}}
\newcommand{\C}{{\mathbb C}}
\newcommand{\Z}{{\mathbb Z}}

\newcommand{\cn}{\mbox {\rm cn}}

\newcommand{\KEll}{{\mathcal K}}
\newcommand{\EEll}{{\mathcal E}}
\newcommand{\Am}{{\mathcal A}}
\newcommand{\SP}{{\mathcal P}}
\newcommand{\en}{\omega}
\newcommand{\En}{\Omega}

\begin{document}
 
 \title [Nonlinear Winter's model]{Quantum resonances and analysis of the survival amplitude in the nonlinear Winter's model}
 
 \author {
 Andrea Sacchetti
 }

\address {
Department of Physics, Informatics and Mathematics, University of Modena and Reggio Emilia, Modena, Italy.
}

\email {andrea.sacchetti@unimore.it}

\date {\today}

\thanks {This work is partially supported by the GNFM-INdAM and the UniMoRe-FIM project  ``Modelli e metodi della Fisica Matematica''.}

\begin {abstract} In this paper we show that the typical effects of quantum resonances, namely, the exponential-type decay of the  survival amplitude, continue to exist even when a nonlinear perturbative term is added to the time-dependent Schr\"odinger equation. \ The difficulty in giving a rigorous and appropriate definition of quantum resonances by means of the notions already used for linear equations is also highlighted. 

\medskip

PACS number(s): {03.65.-w, 03.65.Db, 03.65.Xp, 03.75.Lm, 05.45.-a}

MSC 2020 number(s): {81Qxx}

Declaration of competing interest: the author declares that he has no known competing financial interests or personal 
relationships that could have appeared to influence the work reported in this paper.

Data availability: No data was used for the research described in the article.

\end{abstract}


\maketitle

\section {Introduction}\label {S1}

It is a well known fact that the introduction of a nonlinear perturbation dramatically changes the qualitative behavior of a linear model. \ For example, in classical mechanics, the Duffing equation is a second-order nonlinear equation that describes the motion of a damped oscillator with a cubic perturbation, and the associated dynamical system exhibits a much more complicated behavior than the linear one with (classical) jump resonance phenomena, chaotic dynamics and hysteresis effects \cite {KB}.

Even in quantum mechanics the introduction of a nonlinear term involves a whole range of new problems: from, e.g.,  proving the local and global existence of solutions to the appearance of blow-up phenomena, and most of them have been extensively studied \cite {SS}. \ In contrast, the effect of nonlinear perturbation on quantum resonances is much less understood and studied. 

Within the framework of one-dimensional Schr\"odinger's linear equation
\be
i \partial_t \psi_t = H \psi_t \, , \psi_t \in L^2 (\R ,dx) \, , 
\ee
where $H$ is the linear Schr\"odinger operator, quantum resonances (see \cite {Z} for a review) are associated with metastable states that are not really bound because they usually correspond to states confined by a barrier, through which tunneling occurs. \ The physical effect of quantum resonances can be seen when we consider the time behavior of the \emph {survival amplitude} ${\mathcal A}(t)$ defined as the scalar product between the initial wave function $\psi_0$ and the wave function $\psi_t$ of the quantum system at instant $t$:
\be
\Am (t) :=   \langle \psi_0 , \psi_t \rangle   \, . 
\ee
The \emph {survival probability} is defined as  the
square of the absolute value of the survival amplitude (sometimes in the literature,
with abuse of notation, both objects are named survival probability):
\be
\SP (t) = \left | \Am (t) \right |^2 \, .
\ee

When the pure point spectrum of the  
Schr\"odinger operator $H$ is nonempty and $\psi_0$ is an eigenvector of $H$ then $\SP (t) \equiv 1$ for any $t$; in general, $\SP (t)$ does not go to zero as $t$ goes to infinity if $\Pi_{p} \psi_0 \not= 0$, where $\Pi_{p}$ denotes the projection operator on the pure point eigenspace. \ On the other side, when the pure point spectrum of $H$ is empty or $\Pi_p \psi_0 =0$ we expect to observe an exponentially decreasing behavior for $\SP (t)$ because of the occurrence of quantum resonances (if any). \ In fact, we expect that, after a very short time, the survival amplitude has the following asymptotic behavior \cite {H,NSZ}
\bee
\Am (t) \sim e^{-it E}\label {Formula18}
\eee
when $\psi_0$ is a normalized state approximating the quantum resonance state associated to the resonance energy $E$ such that $\Im E <0$. \ However, we should note that Simon \cite {Simon1} (see also \cite {Ha,HM}) pointed out that exponentially decreasing behavior is dominant for large times only when the Schr\"odinger operator $H$ is not bounded from below. \ In fact, in the case of Schr\"odinger operators $H$ bounded from below, we expect to observe a time decay for the survival amplitude of the form
\bee
\Am (t) \sim e^{-it E}+ b(t) \label {Formula19}
\eee
where the remainder term $b(t)$ becomes the dominant one for very large times. \ Indeed, the time behavior of the survival amplitude is governed by two terms: one term, due to quantum resonances, has exponentially decaying behavior and it is  the dominant one for short times; the second one, due to typical dispersive effects, has an inverse power of the time law and it becomes dominant for longer times. 

The analysis of the problem of the exponential decay rate versus the power
decay rate in the time-dependent survival amplitude $\Am (t)$ has been a research
topic since the '50 and experimental evidence of the deviation from exponential decay has been observed too \cite {W}. \ In the Winter's seminal paper  \cite {Wi} it was numerically conjectured that a transition effect between the two different types of decay starts around a certain instant $t$. \ Recently, a more rigorous analysis of Winter's model (also called single delta-shell model in literature), consisting of a one-dimensional model (see Figure \ref {Fig_1}) with one Dirac's delta potential at $x=a > 0$ and Dirichlet boundary condition at $x = 0$, has been done \cite {AgSa1,AgSa2,de,E,GMV,GL,GVHR,GR2,PVZ,RaK,Wy}. \ Furthermore, Winter-like models, in which a more general singular potential is
considered, have been recently studied, see e.g. \cite {AC,EF}. \ We would remind that the use of a Dirac's $\delta$ function representing a well or barrier potential has been suggested for the first time by E. Fermi \cite {Fermi} in 1936; since then this idea has been widely used to propose simple and solvable models in quantum mechanics \cite {Al}. \ We also obsreve  that the use of Dirac's $\delta$ potentials has been recently introduced to model the effect of nodes in starlike graphs \cite {AN,ACFN,CFN,KND}.
\begin{figure}
\begin{center}
\includegraphics[height=6cm,width=6.25cm]{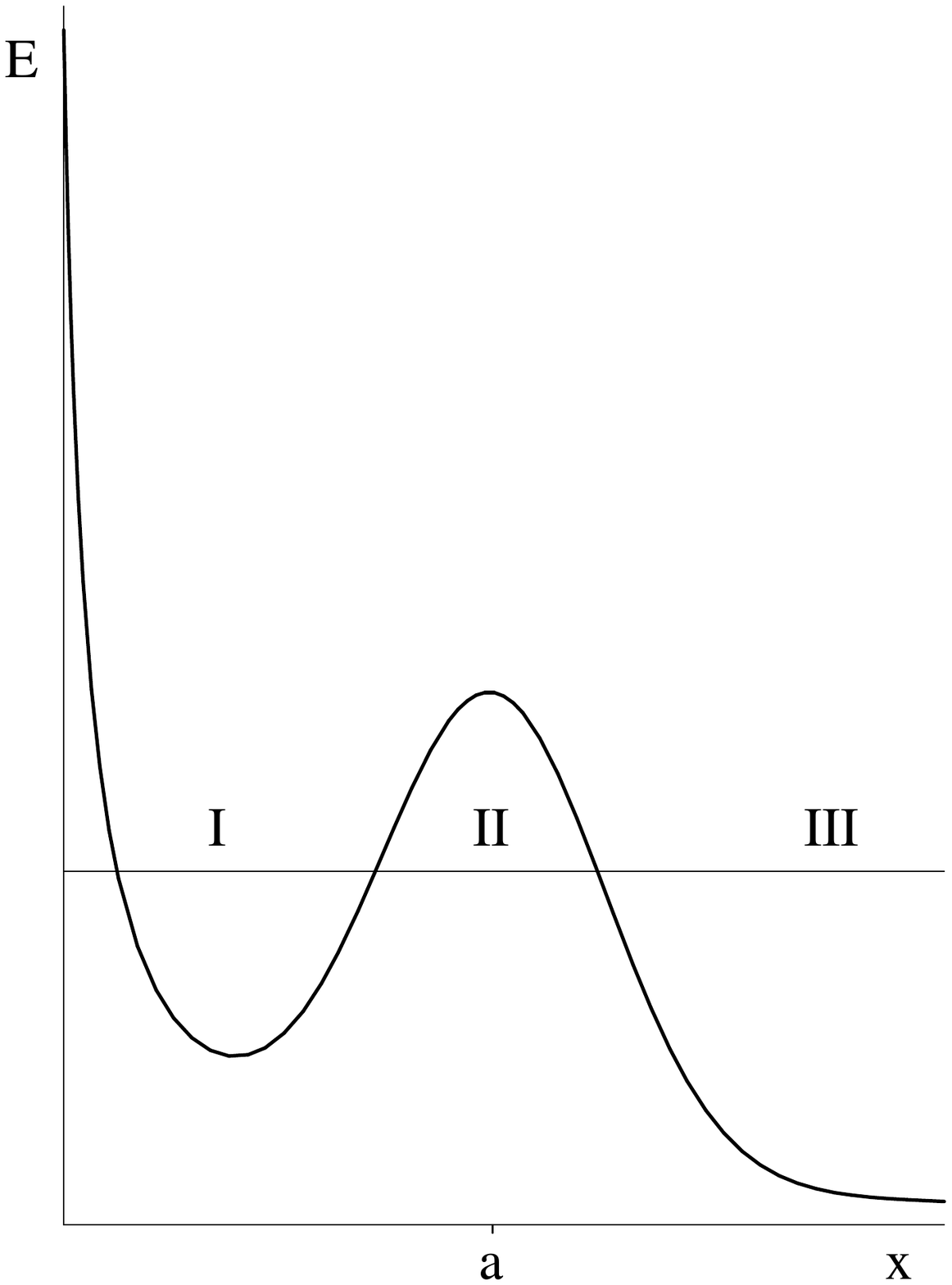}
\includegraphics[height=6cm,width=6.25cm]{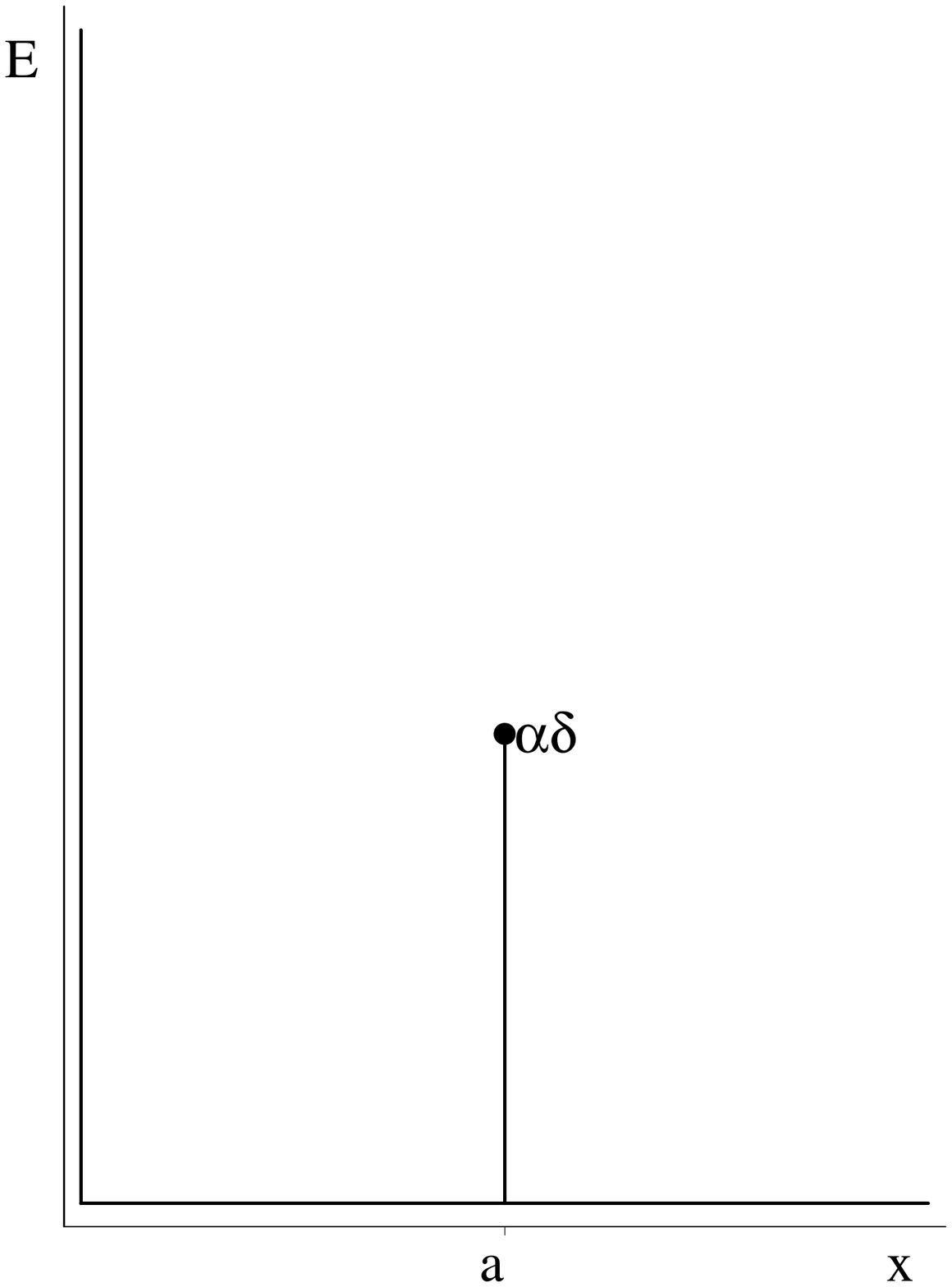}
\caption{A typical potential for the alpha-nucleus tunneling problem is depicted in the left image; for a given energy $E$ the region $I$ indicates the well, while the region $II$ indicates the internal barrier with top at $x=a>0$. \ In fact, in quantum mechanics a particle can move, with a finite nonzero probability, from region $I$ to region $III$ even if these regions are separated by a potential barrier. \ A toy model is depicted in the right image where the barrier is represented by a pointwise $\delta$ Dirac interaction at $x=a$, with strength $\alpha>0$ and an infinite barrier is also present at $x=0$; and this model is usually referred to as Winter's model.}
\label {Fig_1}
\end{center}
\end{figure}

Extending this analysis to the case of nonlinear Schr\"odinger equations (hereafter NLS) raises interesting questions: 

\begin {itemize} 

\item [-] How can we define quantum resonances in the case of nonlinear equations?

\item [-] Is the behavior of the solution $\psi_t$, and in particular of its survival amplitude ${\mathcal A}(t)$, affected by the existence of quantum resonances? And if so, in what way?

\end {itemize}
These questions are still largely unexplored, and the purpose of this paper is to shed some light on this complex problem and to 
try to understand some basic facts from the analysis of a simple explicit model.  \ Indeed, while the extension of the notion of stationary states from linear to nonlinear Schr\"odinger equations is fairly well understood, the extension of quantum resonances and the connection of quantum resonances, if properly defined, with the time decay of the survival amplitude are, on the other hand, far from being fully understood. 

The one-dimensional nonlinear Schr\"odinger equation we consider is the Gross-Pitaevskii one with cubic nonlinearity:
\be
i \partial_t \psi_t = H \psi_t + \eta |\psi_t |^2 \psi_t \, , \eta \in \R \, , \ \psi_t \in L^2 (\R ,dx) \, , 
\ee
One may, in principle, consider the case of quintic nonlinearity, where the nonlinearity term is given by $\eta |\psi_t|^4 \psi_t$, or any power nonlinearity, where the nonlinearity term is given by $\eta |\psi_t|^{2\sigma} \psi_t$ for some $\sigma>0$. \ We don't dwell here on these other models and we simply restrict ourselves to the cubic nonlinearity where $\sigma=1$.

The absolute value of $\psi_t$ in the nonlinear term $\eta |\psi_t|^2 \psi_t$ plays a crucial role in order to prove the conservation of the norm and of the energy. \  In fact, the nonlinearity arising from the square of the absolute value of the wavefunction inhibits superposition principle as well as the analytical properties that connect the quantum resonances to the time behavior of the survival amplitude in linear equations.  \ Furthermore, explicit solutions to the time-independent nonlinear Schr\"odinger equation are only known when the nonlinear term has the form $\eta \psi_t^3$. \ Although it is correct to replace, for the purpose of finding the stationary solutions associated with real values of energy, the term $\eta |\psi_t |^2 \psi_t$ with the term $\eta \psi_t^3$ this substitution is no longer permissible in the study of quantum resonances because in that case complex values of energy must be considered. \ Also in the study of survival amplitude the original term $\eta |\psi_t |^2 \psi_t$ must be retained.

A first crucial question concerns how quantum resonances can be defined for NLS problems and whether this notion makes sense in NLS. \ Several proposals have been given in literature \cite
{CMBS,MCMB,MC,M,RWK,RK,R,SP,SW,WMK,WRK} making use of complex scaling arguments, Siegert's approximation method and scattering coefficient analisys. \ In particular, we critically review these definitions of quantum resonances in NLS and point out that the two methods based on complex scaling and Siegert's approximation have, in our opinion, serious problems. \ On the other hand, the method of defining quantum resonances by scattering coefficient analysis  can be applied in principle, as in the linear model, but the link between the resonances, associated with the maximum values of the scattering coefficient, and time decay of the survival amplitude still remains completely vague. 

Our aim is to extend the analysis of the Winter's model when considering a nonlinear perturbation of the kind $\eta |\psi_t |^2 \psi_t$; where $\eta \in \R$ is the strength of the nonlinear perturbation and it may assume either positive (in which case we speak of \emph {defocusing} or \emph {repulsive} nonlinearity) and negative (in which case we speak of \emph {focusing} or \emph {attractive} nonlinearity) values. \ Nonlinear Winter's problem has been considered by several authors, e.g. \cite {M,RK,R,WMK,WRK} where several attempts to define quantum resonances in such a model are reported. \ In this paper we'll  numerically show that the time decay of the survival amplitude is really affected by the nonlinearity strength. \ In particular, it can be seen by numerical experiments that the typical exponential  decay associated to quantum resonances persists even in this model as long as $\eta$ is not smaller than a critical negative value $\tilde \eta$, and it becomes faster for increasing positive values of $\eta$; if $\eta$ becomes smaller than the critical value $\tilde \eta$ then new stationary states of the nonlinear equation arise and the survival amplitude does not decay. \ Furthermore, we can also conjecture that the quantum resonances obtained in the linear model become stationary states for the nonlinear one when the nonlinearity strength takes the negative threshold value $\tilde \eta $. 

The paper is organized as follows. \ In Section \ref {S2} we consider Winter's linear model in detail; in particular we recall the spectral properties and the expression of the resolvent operator, calculate the quantum resonances, give the expression of the evolution operator and finally, in a numerical experiment, calculate the survival amplitude ${\mathcal A}(t)$. \ Some of these results have been previously given also in other papers \cite {AgSa1,AgSa2,de,E,GMV,GL,GVHR,GR2,PVZ,RaK,Wy}. \ In Section \ref {S3} we then consider  Winter's nonlinear model where we calculate the stationary solutions showing that a bifurcation phenomenon occurs and where we critically review the definition of quantum resonances. \ Finally, in Section \ref {Sec4}, by means of numerical experiments, we show that the survival amplitude does indeed  depend on the nonlinearity strength $\eta$ and  we draw some concluding remarks. \ In Appendix \ref {App1} we provide the technical proof of the Theorem \ref {Teo1}.

Concerning notation: $\| \psi \|$ denotes the usual norm in $L^2$; by $\omega$ we denote the eigenvalue and resonance energy for the linear Winter's model considered in \S \ref {S2}; in \S \ref {S3}, where we study the nonlinear Winter's model, the energy is denoted by $\Omega$ instead of $\omega$.

Normalization to $1$ of the wavefunction $\psi_t$ is assumed valid throughout the whole paper with the exception of \S \ref {Ss34} where we analyze the notion of quantum resonances for NLS. \ We recall that for linear Schr\"odinger equations the value $c>0$ of the normalization condition $\| \psi \| =c$ does not matter because we can always reduce it to $1$ by means of a simple scaling $\psi \to \psi /c$. \ On the other hand, in the case of NLS the scaling $\psi \to \psi /c$ implies that the nonlinearity strength $\eta$ must change as $\eta \to c^2 \eta$. \ Thus, if we decide to change the normalization condition then we must take care of the fact that the nonlinearity strength changes too. \ In \S \ref {Ss34} we consider general solutions to the NLS that are not in $L^2$ and thus we cannot assume the usual normalization condition $\int_0^{+\infty} | \psi (x) |^2 dx =1$. \ We in fact assume a different normalization condition $\int_0^{a} | \psi (x) |^2 dx =1$ and for this reason we denote only in this section the nonlinearity strength by $\Gamma$ instead of $\eta$ as usually do in the rest of the paper.

\section {Analysis of the linear Winter's model}\label {S2}

In this Section we consider the one-dimensional time dependent linear Schr\"odinger equation 
\bee
\left \{
\begin {array}{l}
i\dot \psi_t = H_\alpha \psi_t  \\
\left. \psi_t \right |_{t=0} = \psi_0 
\end {array}
\right.  \, , \ \psi_t \in L^2 (\R^+ ) \, , \ \| \psi_0 \| =1 \, ,  \label {Formula2}
\eee
where
\bee
H_\alpha = - \frac {\partial^2}{\partial x^2} + V  \ \mbox { and } \ 
V(x) = 
\left \{
\begin {array}{ll}
+\infty & \mbox { if } x < 0 \\
\alpha \delta (x-a) & \mbox { if } x \ge 0 
\end {array}
\right. \label {Formula3}
\eee
for some $a>0$ and $\alpha \in \R \cup \{ + \infty \}$, $\alpha \not= 0$. \ For the sake of simplicity we simply denote, when this does not cause misunderstanding, $H_\alpha$ by $H$ when $\alpha \in \R$ and $H_{\infty }$ when $\alpha = + \infty$. \ Similarly, we omit the dependence on $\alpha$ in the other terms, e.g. the resolvent, the kernel of the resolvent operator, and so on, when this fact does not cause misunderstanding.

\subsection {Resolvent and spectrum}\label {Ss21}

Let $\alpha \in \R \setminus \{ 0 \}$, and let 
\bee
\psi (0)=0 \label {Formula4}
\eee
be the Dirichlet boundary condition at $x=0$ and  
\bee
\psi (a-)=\psi (a+) \ \mbox { and } \ \psi' (a+) - \psi' (a-) = \alpha \psi (a) \, , \label {Formula5}
\eee
 be the matching condition at $x=a$ where the Dirac's delta is supported. \ It is well known \cite {Al} that the linear operator $H$ admits a self-adjoint extension (still denoted by $H$) defined on the domain
\be
{D}(H) = \left \{ \psi \in H^{2,1} (\R^+ ) \cap H^{2,2} (\R^+ \setminus \{ a\} ) \ : \  \mbox { (\ref {Formula4}) and (\ref {Formula5}) hold true} \right \} \, . 
\ee

Let 
\be
K_0 (x,k)= \frac {i}{2k} e^{ik|x|}\, , \ \Im k >0 \,
\ee
and let
\be
\Gamma (k) =
\left ( 
\begin {array}{cc}
-  \frac {i}{2k} & -  \frac {i}{2k} e^{ika} \\ 
 -   \frac {i}{2k} e^{ika} &  -   \frac {1}{\alpha}- \frac {i}{2k} 
\end {array}
\right )
\ee
with inverse matrix
\be
\Gamma ^{-1} (k) =\frac {2k}{2ik-\alpha+\alpha e^{2ika}} 
\left (
\begin {array}{cc}
-2k - i\alpha  & i\alpha e^{i k a} \\
 i\alpha e^{i k a} & -i \alpha 
\end {array}
\right )\, . 
\ee
Then the resolvent operator is the integral operator \cite {Al}
\be
\left ( \left [ H -k^2 \right ]^{-1} \phi \right )(x) = \int_{\R^+} K (x,y,k) \phi (y) dy\, , \ \phi \in L^2 (\R^+ )\, , 
\ee
where 
\bee
K (x,y,k) = K_0(x-y,k)  - \frac {1}{4k^2} \sum_{j=1}^4 K_j (x,y,k) \label {Formula6}
\eee
is the kernel with
\be
K_1 (x,y,k) &=& \left [ [\Gamma^{-1} (k) ] \right ]_{1,1} e^{ik (|x|+|y|)} \\
K_2 (x,y,k) &=& \left [ [\Gamma^{-1} (k) ] \right ]_{1,2} e^{ik (|x|+|y-a|)} \\
K_3 (x,y,k) &=& \left [ [\Gamma^{-1} (k) ] \right ]_{2,1} e^{ik (|x-a|+|y|)} \\ 
K_4 (x,y,k) &=& \left [ [ \Gamma^{-1} (k)] \right ]_{2,2} e^{ik (|x-a|+|y-a|)} \, . 
\ee

Concerning the spectrum it follows that 
\bee 
\sigma_{ess} (H) = \sigma_{ac} (H) = [0,+\infty )\, , \label {Formula7}  
\eee
and the eigenvalues, if there, are given by $\en =k^2 <0$ where $k$ is a purely imaginary solution to the equation:
\bee
2ik-\alpha+\alpha e^{2ika} =0 \, , \ \Re k =0 \mbox { and } \  \Im k >0 \, , \label {Formula8}
\eee
obtained from the formula (2.1.13) by \cite {Al} for $\alpha_1 =+\infty$, $\alpha_2 =\alpha$, $y_1=0$ and $y_2 =a$ (according with the notation by \cite {Al}). \ This equation has complex-valued solutions
\bee
k= \frac {i}{2a} \left [ -a \alpha +  W_n \left ( a\alpha  e^{a\alpha } \right )\right ]\label {Formula9}
\eee
where $W_n (z)$ denotes the $n$-th branch of the Lambert special function \cite {C}. \ If we recall that:

\begin {itemize}

\item [i.] $W_0 (z)$ is real-valued if and only if $z \ge -e^{-1}$; in particular

\begin {itemize}

\item [ia.] $W_0(0)=0$,

\item [ib.] $W_0(-e^{-1})=-1$, 

\item [ic.] $
W_0 (z) \in [-1,0] \ \Leftrightarrow \ z \in [-e^{-1} ,0 ]$,

\item [id.] $W_0 (z) >0  \ \Leftrightarrow \ z>0$,

\item [ie.] the branch cut for $W_0 (z)$ is the line $(-\infty , - e^{-1} ]$;

\end {itemize}

\item [ii.] $W_{-1} (z)$ is real-valued if and only if $-e^{-1} \le z < 0$, and it takes values in the interval $[-1,0)$;
 
\item [iii.] $W_n (z)$ has not zero imaginary part for any $z \in \R$ and any $n \in \Z \setminus \{ 0,-1\} $, furthermore the branch cut for $W_{n} (z)$, $n \not= 0$, is the line $(-\infty , 0 ]$;

\end {itemize}
then solutions (\ref {Formula9}) are purely imaginary and such that $\Im k >0$ only if $n=0$ and $a\alpha < -1$. \ In conclusion, we have proved that

\begin {proposition} \label {Prop1}
If $a\alpha < -1$ then the discrete spectrum of $H$ is not empty and it consists of just one negative real-valued eigenvalue
\bee
\en =- \left [ \frac {1}{2a} \left [ -a\alpha + W_0 \left ( a\alpha  e^{a \alpha } \right )\right ] \right ]^2 \, . \label {Formula10}
\eee
If $a\alpha \ge -1$ then the discrete spectrum of $H$ is empty.
\end {proposition}

\begin {remark} \label {Nota1}
If $\alpha =+\infty$ then condition (\ref {Formula5}) becomes the Dirichlet boundary condition $\psi (a)=0$ and in such a case the spectrum of $H_\infty$ is purely discrete with eigenvalues
\be
\en_{\infty ,m} = \left ( \frac {m\pi}{a} \right )^2 \, , \ m=1,2 ,\ldots \,  
\ee
and associated normalized eigenvectors 
\bee
\psi_{\infty ,m} (x) = \sqrt {\frac {2}{a}} \sin \left ( \frac {m\pi x}{a} \right ) \chi_{(0,a)}(x) \, , \label {Formula11}
\eee
where
\be
\chi_A (x) = 
\left \{
\begin {array}{ll}
0 & \mbox { if } x \notin A \\ 
1 & \mbox { if } x \in A 
\end {array}
\right. \, .
\ee
\end {remark}

\subsection {Barrier quantum resonances} \label {Ss22}

In the case of a repulsive $\delta$ interaction at $x=a$, i.e. for $\alpha >0$, the discrete spectrum of $H$ is empty. \ However, quantum resonances may occur. 

Quantum resonances for linear Schr\"odinger operators may be defined in several ways (see \cite {E,Mo,RFC,Simon1,Simon2,Z} for a review); here, we identify quantum resonances with complex poles of the kernel of the analytic continuation of the resolvent operator. \ More precisely, let ${\mathcal D}$ be a dense subset of $L^2 (\R^+)$, then for any $\varphi \in {\mathcal D}$ the function $\en \in \C \to \langle \varphi, [H-\en ]^{-1} \varphi \rangle$ has a meromorphic continuation from the upper half-plane $\Im \en >0$ to the lower half-plane $\Im \en <0$; resonances are the complex poles of such an analytic extension.  

By means of such a definition and by making use of the resolvent kernel formula (\ref {Formula6}) then quantum resonances $\en=k^2$ of $H$ are associated to the complex solutions $k$ to (\ref {Formula8}) such that $\mbox {arg} (k) \in (-\pi /4 , 0)$. \ Hence, it follows that

\begin {proposition} \label {Prop2}
If $\alpha >0$ the discrete spectrum of $H$ is empty and $H$ admits a family of quantum resonances $\en_m = k_m^2$ where 
\bee
k_m =  i w_m \, , \ w_m = \frac {1}{2a} \left [  - a\alpha + W_{-m} \left ( {a\alpha} e^{ {a\alpha}} \right ) \right ]  \, ,  \ m=1,2 ,\ldots \, . \label {Formula12}
\eee
\end {proposition}

In Table \ref {Tavola1} and in Figure \ref {Fig_2} we collect the values of the first 10 resonances $\omega_m$, $m=1,\ldots , 10$, for different values of $\alpha$.

\begin{table}
\begin{center}
\begin{tabular}{|l||c|c||c|c||c|c||c|} 
\hline
\multicolumn{1}{|l||}{ } & \multicolumn{2}{|c||}{$\alpha =1$} & \multicolumn{2}{|c||}{$\alpha =5$} & \multicolumn{2}{|c||}{$\alpha =10$} & \multicolumn{1}{|c|}{$\alpha =+\infty $}
\\ \hline
$m $          &  $\Re \en_m $   & $\Im \en_m $   &  $\Re \en_m $   & $\Im \en_m $  &  $\Re \en_m $   & $\Im \en_m$ & $\en_{\infty ,m}$ \\ \hline \hline 
$1$      &  $ 4.70 $ & $ -3.52 $ & $ 7.31$  & $-0.96$ & $ 8.28 $  & $ -0.38$  & $ 9.87 $ \\ \hline
$2$      &  $ 28.10 $ & $ -13.01 $ & $31.98$ & $-5.00$ & $ 34.08 $  & $ -2.41 $  & $ 39.48 $  \\ \hline
$3$      &  $ 71.69 $ & $ -24.46 $ & $76.06$ & $ -11.14$ & $ 78.75 $  & $ -6.18 $  & $ 88.83 $  \\ \hline
$4$      &  $ 135.22 $ & $ -37.08 $ & $139.88$ & $-18.57$ & $ 142.87 $  & $ -11.24 $  & $ 157.91 $  \\ \hline
$5$      &  $ 218.60 $ & $ -50.55 $ & $223.47$ & $-26.89$ & $ 226.64 $  & $ -17.24 $  & $ 246.74 $  \\ \hline
$6$      &  $ 321.81 $ & $ -64.68 $ & $326.83$ & $-35.91$ & $ 330.13 $  & $ -23.98 $  & $ 355.31 $  \\ \hline
$7$      &  $ 444.81 $ & $ -79.36 $ & $449.97$ & $-45.49$ &  $ 453.35 $  & $ -31.30 $  & $ 483.61 $  \\ \hline
$8$      &  $ 587.58 $ & $ -94.50 $ & $592.86$ & $-55.55$ &  $ 596.30 $  & $ -39.12 $  & $ 631.65 $  \\ \hline
$9$      &  $ 750.13 $ & $ -110.05 $ & $755.50$ & $-66.02$ & $ 759.01 $  & $ -47.36 $ & $ 799.44 $  \\ \hline
$10$     &  $ 932.44 $ & $ -125.96 $ & $937.90$ & $-76.85$ & $ 941.46 $  & $ -55.98 $  & $ 986.96 $  \\ \hline
\end{tabular}
\caption{Table of values of quantum resonances $\en_m$ of $H$ with repulsive singular potential for different values of the 
strength $\alpha$; for argument's sake sake we fix  the units such that $a=1$. \ In the last column we collect the values $\en_{\infty ,m} =(\pi m/a)^2$ corresponding to the real-valued eigenvalues obtained in the case of two infinite barriers at $x=0$ and $x=a$.}
\label{Tavola1}
\end{center}
\end {table}
                               
\begin{figure}
\begin{center}
\includegraphics[height=7cm,width=9cm]{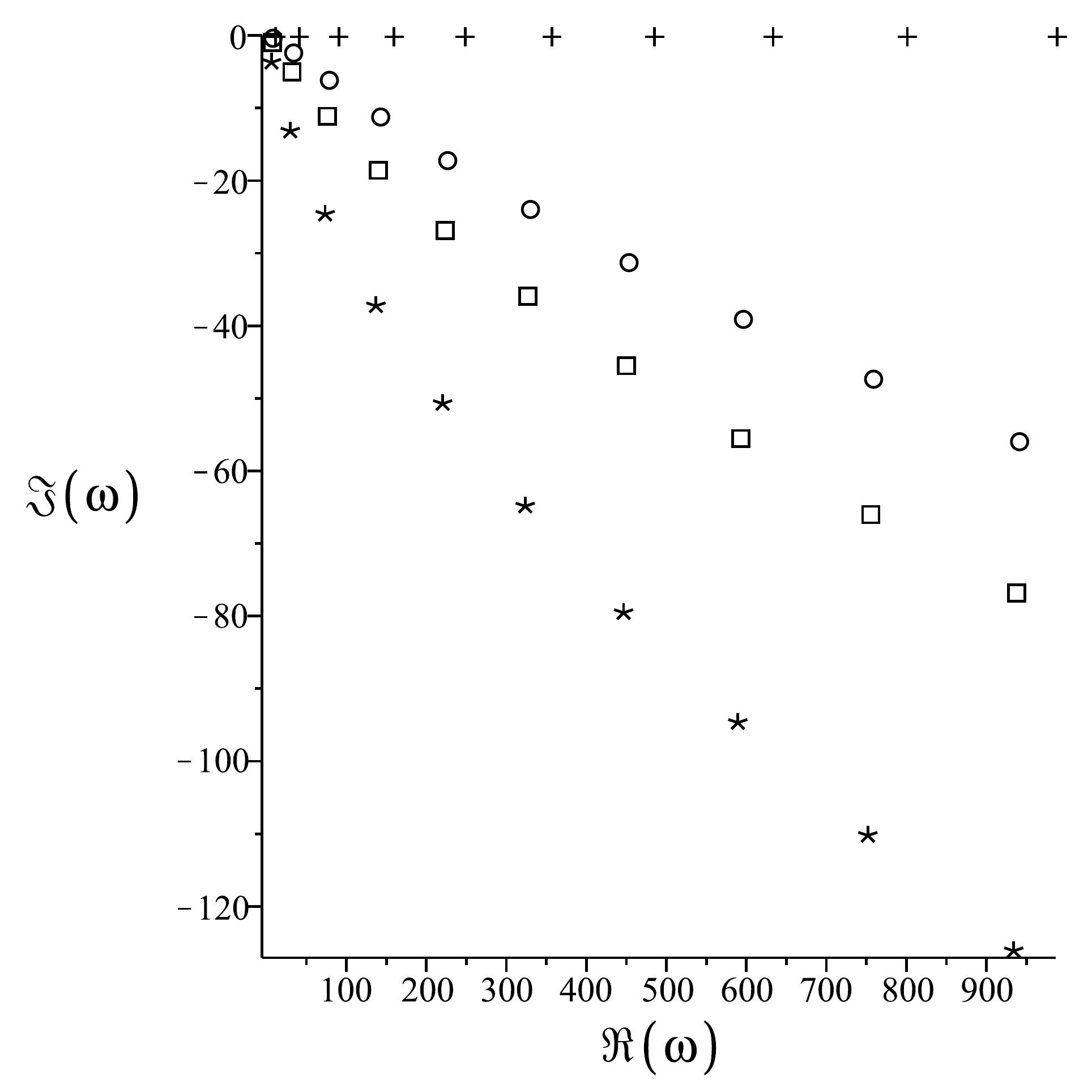}
\caption{Plot of values of quantum resonances $\en_m$ collected in Table \ref {Tavola1} for different values of the 
strength $\alpha$ (asterisk symbols correspond to $\alpha =1$, box symbols correspond to $\alpha =5$, circle symbols correspond to $\alpha = 10$ and finally cross symbols correspond to $\alpha =+\infty$). \ For argument's sake we fix $a=1$.}
\label {Fig_2}
\end{center}
\end{figure}

\begin {remark} \label {Nota2} As $\alpha >0$ increases, the imaginary part of the quantum resonances becomes smaller and smaller and we speak of \emph {narrow resonances}. \ In fact, one can check that the quantum resonances $\en_m $ go to $\en_{\infty ,m}$ for increasing values of $\alpha$ and $a$ fixed:
\be
\lim_{\alpha \to + \infty } \en_m =\en_{\infty ,m} \ \mbox { for any fixed } a>0\, .
\ee
Indeed, we remind that (see formula (4.20) by \cite {C}) 
\bee
W_m (z) \sim \ln z + 2\pi i m  - \ln \left ( 2\pi i m + \ln z \right )\ \mbox { for large } z\, , \label {Formula13}
\eee
and then 
\be
k_m =  \frac {1}{2ia} \left [  a\alpha - W_{-m} \left ( a\alpha e^{ a\alpha} \right ) \right ] 
 \sim  \frac {m\pi }{a} \ \mbox { as } \ \alpha \to + \infty \, . 
\ee
\end {remark}

\begin {remark} \label {Nota3}
Quantum resonances can be also defined as the complex values $\en=k^2$ such that the associated solution to the equation $H\psi =\en \psi$ satisfies the outgoing condition $\psi (x) \sim e^{ikx}$ when $x$ goes to plus infinity (i.e.: Siegert's approximation method \cite {DM,R}). \ In such a case we have to solve the differential equation 
\bee
- \psi '' =k^2 \psi \, , \ x \in (0,a) \cup (a,+\infty ) \label {Formula12Bis}
\eee
with conditions (\ref {Formula4}) and (\ref {Formula5}) and the outgoing condition 
\be
\psi (x) = e^{ikx} \, , \ \forall x >a.
\ee
A straightforward calculation proves that these  conditions are fulfilled provided that $k$ is a solution to (\ref {Formula8}). 
\end {remark}

\begin {remark} \label {Nota4}
Another way to define quantum resonances consists to find the values $\en =k^2$ such that the scattering coefficient becomes singular. \ That is, let
\be
\psi (x) = 
\left \{
\begin {array}{ll}
C_1 \sin (kx+\varphi_1 ) & \mbox { if } x \in (0,a) \\ 
C_2 \sin (kx+\varphi_2 ) & \mbox { if } x > a 
\end {array}
\right. 
\ee
be the solution to (\ref {Formula12Bis}). \ Dirichlet condition $\psi (0)=0$ implies that $\varphi_1 =0$; while the matching conditions (\ref {Formula5}) at $x=a$ imply that
\be
\left \{ 
\begin {array}{lcl}
C_2 \sin (ka+\varphi_2 ) &=& C_1 \sin (ka) \\ 
k C_2 \cos (ka+\varphi_2 ) &=& k C_1 \cos (ka) +\alpha C_1 \sin (ka) 
\end {array}
\right. 
\ee
from which it follows that
\be
k^2 C_2^2 = k^2 C_1^2 + \alpha^2 C_1^2 \sin^2(ka) + k\alpha C_1^2 \sin (2ka) \, . 
\ee
If we define in the Winter's model the scattering coefficient as
\be
S(\en ) = \frac {C_1^2}{C_2^2} = \frac {k^2}{k^2  + \alpha^2 \sin^2(ka) + k\alpha \sin (2ka)}\, , \ \en =k^2 \, , 
\ee
then it has complex poles $k_m$ given by (\ref {Formula12}) and the function $S(\omega )$, for $\omega \in [0,+\infty )$, has a sequence of maximum values (see Figure \ref {Fig_3})  in a neighborhood of $\Re k_m^2$.
\end {remark}

\begin{figure}
\begin{center}
\includegraphics[height=6cm,width=8cm]{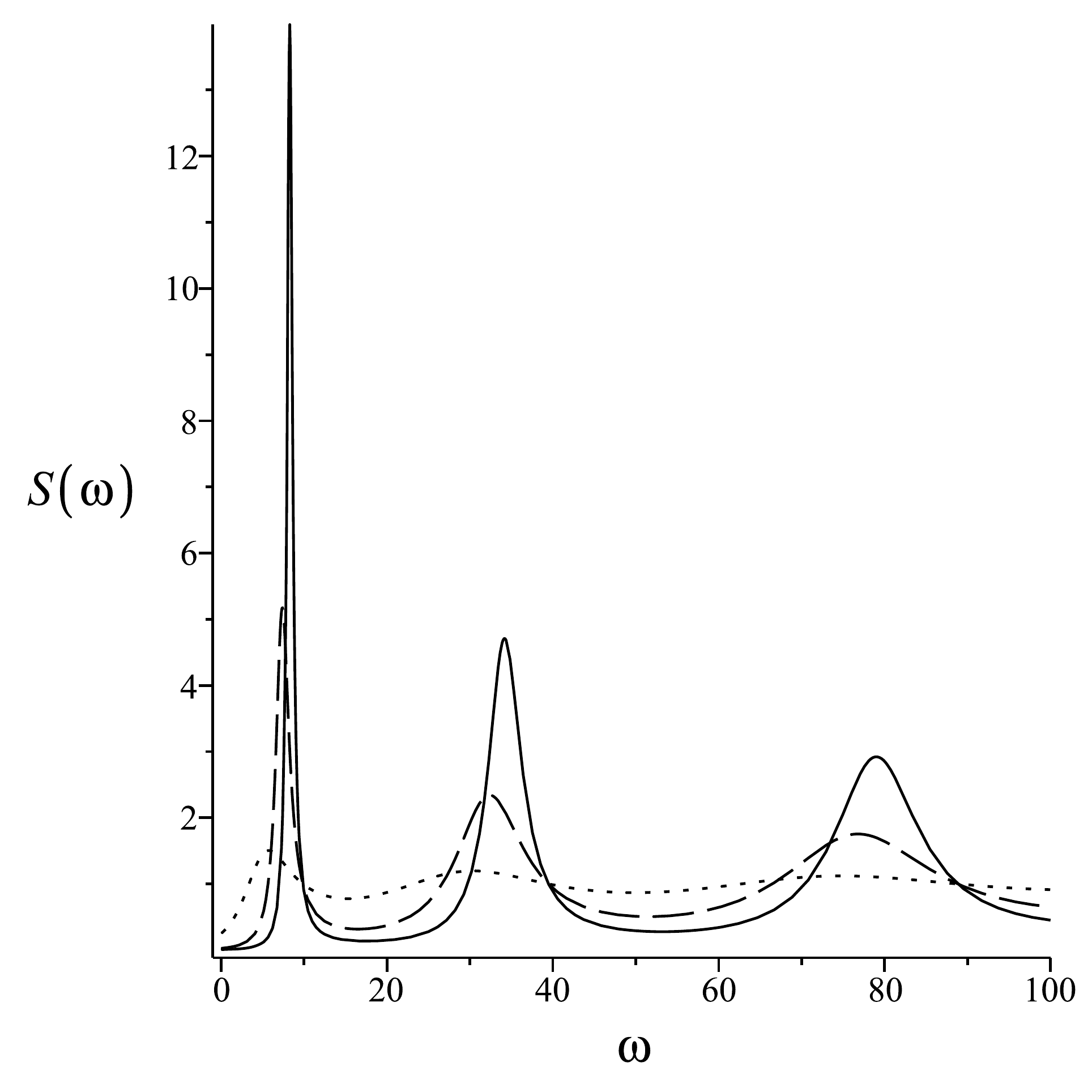}
\caption{Plot of the scattering coefficient $S (\en )$  for $\alpha =+1$ (dot line), $\alpha =+5$ (broken line) and $\alpha =+10$ (full line).}
\label {Fig_3}
\end{center}
\end{figure}

\subsection {Evolution operator} \label {Ss23} Solution $\psi_t (x) \in L^2 (\R^+ ) $ to the time-dependent linear Schr\"odinger equation (\ref {Formula2}) is given by $\psi_t =e^{-itH} \psi_0$, where $e^{-itH}$ is the evolution operator associated to the self-adjoint operator $H$. \ Expression of the evolution operator can be recovered from the resolvent operator by making use of arguments similar to the ones used by \cite {KS}. \ Indeed, the evolution operator is an integral operator 
\bee
\left [ e^{-itH} \psi_0 \right ] (x) = \int_{\R^+} U(x,y,t) \psi_0 (y) dy \label {Formula15}
\eee
where the kernel $U(x,y,t)$ has the form
\be
U (x,y,t) = - \frac {i}{\pi} \int_{\R +i0} k e^{-ik^2 t} K (x,y,k ) dk = U_0 (x,y,t)  + \sum_{j=1}^4 U_j (x,y,t) 
\ee
where
\be
U_0 (x,y,t) = - \frac {i}{\pi} \int_{\R +i0} k e^{-ik^2 t} K_0 (x-y,k ) dk = \frac {1}{\sqrt {4\pi i t}}e^{i|x-y|^2 /4t}
\ee
and
\bee
U_j (x,y,t) 
= \frac {i}{4\pi} \int_{\R +i0} \frac {1}{k} e^{-ik^2 t} K_j (x,y,k ) dk\, . 
\label {Formula16}
\eee

In order to apply formula (\ref {Formula15}) in numerical experiments we have to numerically compute the above integrals (\ref {Formula16}). \ Here, we propose a faster way to compute the kernel $U(x,y,t)$ by means of a convergent series. \ The following result, which proof is postponed in the Appendix, holds true

\begin {theorem} \label {Teo1} Let
\be
\begin {array}{lcl}
e_n^1 &:=& e_n^1 (x,y,t) = (2an+|x|+|y|)/{2 \sqrt {t}}  \\
e_n^2 &:=& e_n^2 (x,y,t) =  (2an+|x|+|y-a|-a)/{2 \sqrt {t}}  \\
e_n^3 &:=& e_n^3 (x,y,t) =  (2an+|x-a|+|y|-a)/{2 \sqrt {t}}\\
e_n^4 &:=& e_n^4 (x,y,t) =  (2an+|x-a|+|y-a|-2a)/{2 \sqrt {t}} 
\end {array}
\, , \  n=0,1,2, \ldots \, , 
\ee
and 
\be
f_n^j = e_n^j + i\alpha \sqrt {t} /2 \ \mbox { and } \ g_n^j = e^{i(e_n^j)^2} \, , \ j=1,2,3,4 \, . 
\ee
Let
\be
U_0 := U_0 (x,y,t)  = \frac {1}{\sqrt {4\pi i t}}e^{i|x-y|^2 /4t}\, , 
\ee
\be 
V_0 := V_0 (x,y,t) =  -\frac {1}{\sqrt {8\pi}} 
\left ( {it}/{2} \right )^{-1/2} e^{-i(f_0^1)^2/2} g_0^1 D_{0} ((1-i)f_0^1) 
\ee
and let
\be
v_n &:=& v_n (x,y,t) = - e^{-i(f_n^1)^2/2} g_n^1 D_{-n} ((1-i)f_n^1) + e^{-i(f_{n}^2)^2/2}  g_{n}^2 D_{-n} ((1-i)f_{n}^2) + \\
&& \ \ +  e^{-i(f_{n}^3)^2/2}  g_{n}^3 D_{-n} ((1-i)f_{n}^3) - e^{-i(f_{n}^4)^2/2}  g_{n}^4 D_{-n} ((1-i)f_{n}^4) \, , 
\ee
where $D_n (z)$ denotes the parabolic cylinder function. \ Then
\bee
U (x,y,t) = U_0 + V_0 + \frac {1}{\sqrt {8\pi} } \sum_{n=1}^\infty  \alpha^{n} (it/2)^{(n-1)/2} v_n\, .\label {Formula17}
\eee
\end {theorem}

\begin {remark} \label {Nota5}
If we recall the following properties of the parabolic cylinder function \cite {AS} 

\begin {itemize}

\item [i.] $D_0 (z)=e^{-z^2/4}$;

\item [ii.] $D_{-1} (z) = e^{x^2/4}\sqrt {\frac {\pi}{2}} \mbox {\rm erfc} \left ( \frac {z}{\sqrt {2}}\right )$;

\item [iii.] $D_{-m} (z) = \frac {z D_{1-m} (z) - D_{2-m} (z)}{1-m}$, $m=2,3,\ldots $;

\end {itemize}
then it follows that
\be
V_0 (x,y,t) &=&  -\frac {1}{\sqrt {4\pi i\pi}} 
 e^{i(|x|+|y|)^2/4t}
\ee
and that
\be
e^{-i(f_1^j)^2/2} g_1^j D_{-1} ((1-i)f_1^j) = \sqrt {\frac {\pi}{2}} e^{-i(f_1^j)^2} g_1^j \mbox {\rm erfc} \left ( \frac {(1-i)f_1^j}{\sqrt {2}}\right ) 
\ee
Hence, by induction, terms $v_n$  can be computed by means of the error function $\mbox {\rm erfc}$. \ In particular, by means of the asymptotic expansion (7.1.23) \cite {AS} one can check that
\be
D_{-m}(z) \sim z^{-m} e^{-z^2/4}\ \mbox { as } z \to \infty \, . 
\ee
In conclusion, it follows that
\be
v_n \sim - \frac {g_n^1}{[(1-i)f_n^1]^n}+ \frac {g_n^2}{[(1-i)f_n^2]^n}+ \frac {g_n^3}{[(1-i)f_n^3]^n}- \frac {g_n^4}{[(1-i)f_n^4]^n} \sim (na/\sqrt {t})^{-n}
\ee
for large $n$ and thus the series (\ref {Formula17}) rapidly converges for any $t$ and $\alpha$.
\end {remark}

\begin {remark} \label {Nota6}
By means of a straightforward calculation one can check that
\be
U(x,y,t)=0 \ \mbox { when } \ xy \le 0 \, .
\ee
Indeed, if, for instance, $x\le 0$ and $y \ge 0$ then 
\be
V_0 (x,y,t) =  -\frac {1}{\sqrt {4\pi i\pi}} 
 e^{i(|x|+|y|)^2/4t}  = -\frac {1}{\sqrt {4\pi i\pi}} 
 e^{i(-x+y)^2/4t} = - U_0 (x,y,t)
 \ee
 and for any $n=0,1,2, \ldots$
 \be
\begin {array}{lcl}
e_n^1 &=& (2an-x+y)/{2 \sqrt {t}} \\
e_n^2 &=& (2an-x+|y-a|-a)/{2 \sqrt {t}} \\
e_n^3 &=& (2an-x+y)/{2 \sqrt {t}}=e_n^1\\
e_n^4 &=& (2an-x+|y-a|-a)/{2 \sqrt {t}}=e_n^2 
\end {array}
\,  ,
\ee
from which it follows that $v_n=0$ for any $n$.
\end {remark}

\subsection {Survival amplitude} \label {Ss24}

Let $\psi_t (x)$ be the solution to (\ref {Formula2}) with initial condition $\psi_0 (x)$. \ We define \emph {survival amplitude} the scalar product between these two vectors, that is
\be
\Am (t) :=   \langle \psi_0 , \psi_t \rangle   \, . 
\ee
In order to discuss the exponential    behavior (\ref {Formula19}) in the Winter's model associated to the quantum resonances $k_m$ given by Proposition \ref {Prop2} we consider the following experiment: let us choose $\psi_0$ coinciding with the ground state wavefunction of $H_\infty$  
\be
 \psi_{\infty ,1} (x) = \sqrt {\frac {2}{a}} \sin \left ( \frac {\pi x}{a} \right )\chi_{[0,a]}(x)\, , 
\ee
where $\chi_{[0,a]}(x)$ is the step function in the interval $[0,a]$, with associated eigenvalue $\en_{\infty ,1} = \frac {\pi^2}{a^2}$. \ 
Then we compute the survival amplitude $
\Am (t) := \langle \psi_0 , \psi_t \rangle $ where $\psi_t = e^{-iH t} \psi_0$ and $\psi_0 = \psi_{\infty ,1}$, for different values of $\alpha$ (e.g. $\alpha =1$, $\alpha =10$ and $\alpha =100$). \ Numerical computation of $\psi_t$, and then of $\Am (t)$, could be done by making use of (\ref {Formula15}) where the kernel $U(x,y,t)$ is given by means of the integrals (\ref {Formula16}) or, more quickly and easily, by making use of the convergence series given in Theorem \ref {Teo1}. \ In fact, because of the particular choice  of the initial wavefunction $\psi_0$ we don't necessarily need to make use of these numerical tools but, in order to compute the survival amplitude, we could  make use of the following Theorem.

\begin {theorem} \label {Teo2} Let $k_m$ be the complex-valued solutions to (\ref {Formula8}) given in Proposition \ref {Prop2}; let 
\be
\beta_m = \left \{ 
\begin {array}{ll}
0 & \mbox { if } |\Im k_m | > |\Re k_m | \\
\frac 12  & \mbox { if } |\Im k_m | = |\Re k_m | \\
1 & \mbox { if } |\Im k_m | < |\Re k_m | 
\end {array}
\right. \, ,
\ee
and 
\be
a_1 &=&  - \frac 
{(1+i)\sqrt {2}a^3}
{4 (1+a\alpha )^2 \pi^{9/2} }
\left [ -8  (1+a\alpha )^2 +\pi^2 (a^2 \alpha^2+2a \alpha + 5 ) \right ] \\
c_m &=& 2\pi i q_m \, , \ q_m = \frac {a\pi k_m}{1+a(\alpha -2 i k_m)} 
\left [ \frac { 1+e^{ik_m a}}{\pi^2 - k_m^2 a^2} \right ]^2 
\ee
then
\be
 \Am (t) = \langle \psi_0 , \psi_t \rangle  = a_1 t^{-3/2} - \sum_{m=1}^\infty \beta_m c_m e^{-ik_m^2 t} + O (t^{-5/2}) \ \mbox { as } \ t\to + \infty \, . 
\ee
\end {theorem}

\begin {proof}
In order to compute the survival amplitude we follow the line introduced by \cite {Sacchetti2016,Sacchetti2017}. \ In particular, we have that
\be
\langle \psi_0 , \psi_t \rangle =  \langle \psi_0 , e^{-itH} \psi_0 \rangle = 
f_0 (t) + f_\alpha (t) \, ,\ f_\alpha (t) =\sum_{j=1}^4 f_j (t)\, , 
\ee
where
\be
f_0 (t) &=&  \int_{\R} Q_0 (k) e^{-i k^2t  } dk\, ,\ Q_0 (k) = \frac {k}{\pi i} \int_{\R} \int_{\R} \overline {\psi_0 (x)} \psi_0 (y) K_0(x-y,k) dy \, dx 
\ee
is the evolution term associated to the free Laplacian; and 
\be
f_j (t) &=&  \int_{\R} Q_j (k)  e^{-i k^2 t  } dk\, ,\ Q_j (k) = -\frac {1}{4k \pi i} \int_{\R} \int_{\R} \overline {\psi_0 (x)} \psi_0 (y) K_j(x,y,k) dy \, dx \, . 
\ee
A straightforward calculation gives that
\be
Q_0 (k) 
= \frac {a\left [2i\pi^2 (1+e^{ika})+ka (\pi^2-k^2a^2)\right ]}{\pi i (k^2a^2 -\pi^2 )^2}\, . 
\ee
Concerning the terms $Q_j(k)$ for $j=1,2,3,4$ we have that
\be
Q_1 (k) &=&  -\frac {1}{4k\pi i} \int_{\R} \int_{\R} \overline {\psi_0 (x)} \psi_0 (y) K_1(x,y,k) dy dx \\ 
&=& -\frac {1}{2ak\pi i} \left [ \Gamma^{-1}  (k)\right ]_{1,1} \int_{0}^a \int_{0}^a \sin \left (\frac {\pi x}{a} \right ) \sin \left (\frac {\pi y}{a} \right ) e^{ik(|x|+|y|)}  dy dx \\
&=& -\frac {1}{2ak\pi i} \left [ \Gamma^{-1}  (k) \right ]_{1,1} \left [ \frac {\pi a (1+e^{ika})}{\pi^2-k^2a^2} \right ]^2 
\ee
and similarly
\be
Q_2 (k) &=&  -\frac {1}{2ak \pi i} \left [  \Gamma^{-1}  (k) \right ]_{1,2} \left [ \frac {\pi a (1+e^{ika})}{\pi^2-k^2a^2} \right ]^2\\
Q_3 (k) &=&  -\frac {1}{2ak\pi i} \left [ \Gamma^{-1}  (k) \right ]_{2,1} \left [ \frac {\pi a (1+e^{ika})}{\pi^2-k^2a^2} \right ]^2\\
Q_4 (k) &=&  -\frac {1}{2ak\pi i} \left [  \Gamma^{-1}  (k) \right ]_{2,2} \left [ \frac {\pi a (1+e^{ika})}{\pi^2-k^2a^2} \right ]^2
\ee
Now, let
\be
Q(k) = \sum_{j=1}^4 Q_j (k) 
= \frac {q(k)}{2k+i\alpha-i\alpha e^{2ika}}
\ee
where
\be
q(k) = -2a \pi \left [ \frac { (1+e^{ika})}{\pi^2-k^2a^2} \right ]^2 (k+i\alpha - i \alpha e^{ika} )
\ee

From the Cauchy Theorem as applied in Lemma 3 by \cite {Sacchetti2016} and the Watson's Lemma stated by \S 43.3 \cite {Fedoriuk} it follows that
\be
f_0 (t) &=& \int_{\R} Q_0 (k) e^{-ik^2 t} dk \\
&=& \frac {2 \sqrt {2} a (1-i)}{\pi^{5/2} \sqrt {t}} + \frac {a^3\sqrt{2}}{4 \pi^{9/2} t^{3/2}} \left [ (\pi^2 - 8)(1+i) \right ] +O(t^{-5/2}) \ \mbox { as } \ t \to + \infty \, .
\ee

Eventually, we have to calculate
\be
f_\alpha (t) &=& \int_{\R} \frac {q(k)}{2k+i\alpha-i\alpha e^{2ika}} e^{-k^2it } dk  \\
&=& e^{-i\pi /4}  \int_{\R} 
Q \left ( e^{-i\pi /4} \rho \right )  e^{-\rho^2 t } d\rho -  \sum_{m=1} \beta_m 2\pi i \mbox {Res} \left [ \frac {q(k) e^{-k^2it }}{2k+i\alpha -i\alpha e^{2ika}},k_m \right ]\, ,
\ee
from the Residue's Theorem; then, again Watson's Lemma gives that
\be
\int_{\R} 
Q \left ( e^{-i\pi /4} \rho \right )  e^{-\rho^2 t } d\rho =
t^{-1/2}\Gamma (1/2) d_0 + \frac 12 t^{-{3/2}} \Gamma (3/2) d_1+ O(t^{-5/2}) \ \mbox { as } \ t \to + \infty \, , 
\ee
where 
\be
d_0 &=& -\frac {4a}{\pi^3 }\\
d_1 &=& -\frac {2ia^3}{(1+a\alpha )^2\pi^5} \left [ -8 (1+a\alpha )^2 + \pi^2 \left ( a^2\alpha^2 + 2 a \alpha +3  \right ) \right ]\, . 
\ee

Concerning the calculus of the residues it follows that 
\be
\mbox {Res} \left [ \frac {q(k) e^{-k^2it }}{2k+i\alpha -i\alpha e^{2ika}},k_m \right ] 
={q_m}e^{-k^2_m it }  
\ee
where
\be
q_m=  \frac {q(k_m) }{2+2a(\alpha -2ik_m)}\, . 
\ee
Then, Theorem \ref {Teo2} follows. 
\end {proof}

\begin {remark} \label {Nota7}
From Theorem \ref {Teo2} it follows that the dominant terms of the survival amplitude for large times are
\be
a_1 t^{-3/2} \ \mbox { and } \ c_1 e^{ \Im (\en_1) t}
\ee
since $\Im \omega_m < \Im \omega_1 <0$ for $m=2,3,\ldots $. \ In particular, the first term is the dominant one when $t$ goes to infinity since $\Im \en_1 < 0$; while the second one is the dominant one for 
\be
t \le \frac {3}{2\Im \en_1} W\left ( -1 , \frac {2\Im \en_1}{3} \left ( \frac {a_1}{c_1} \right )^{2/3} \right ) \sim |\Im \en_1 |^{-1} \left | \ln (|\Im \en_1 |)\right | \ \mbox { for } \ |\Im \en_1 | \ll 1 \, . 
\ee
\end {remark}

In Figure \ref {Fig_4} we plot the absolute value of the survival amplitude $\left | \Am (t) \right |$ for different values of $\alpha$ (where we fix the units such that $a=1$) where $\psi_0 = \psi_{\infty ,1}$. \ In order to compare the results obtained by formula 
\bee
\Am (t) \sim a_1 t^{-3/2} + c_1 e^{-i k_1^2 t} \label {Formula20}
\eee
with the ones obtained when $\psi_t$ is computed by formula (\ref {Formula15}) and Theorem \ref {Teo1} we compute in Table \ref {Tavola2}  the maximum $\Delta$ of the absolute value of their difference for $t\in [0.5,5]$ and we see that this difference turn out to be very small; thus the two results fully agree. \ We have to point out that the computation of $\Am (t)$ by means of formula (\ref {Formula15}) can be done, in principle, for any time $t$ but it is much more time-consuming than the simple formula (\ref {Formula20}) that properly works when $t$ is not too small (e.g. $t \ge 0.5$) in the considered experiment.

\begin{table}
\begin{center}
\begin{tabular}{|c|c|c|c|c|c|} 
\hline
$\alpha $  &  $k_1 $   & $ \en_1 $   &  $c_1 $   
& $a_1 $ & $\Delta$   \\ \hline \hline
$1 $  &  $ 2.2986-0.7660i $   & $4.6966-3.5216i $  
& $ -1.1943+0.4624i$  &  $ -0.1011\cdot 10^{-1} (1+i)$  & $0.96 \cdot 10^{-2}$   \\ \hline
$10 $  &  $2.8776-0.0665i $   & $8.2766-0.3828i $   &  $ -0.9898+0.0303i$   
& $-0.3331 \cdot 10^{-3} (1+i) $  & $0.72 \cdot 10^{-3}$   \\ \hline
$20 $  &  $2.9958-0.0205i $   & $8.9742-0.1231i $   &  $-0.9950+0.0085i $   
& $-0.9166 \cdot 10^{-4} (1+i) $  & $0.22 \cdot 10^{-3}$   \\ \hline
$40 $  &  $ 3.0655-0.0057i$   & $9.3974-0.0347i $   &  $-0.9983+0.0021i $   
& $-0.2405 \cdot 10^{-4} (1+i) $  & $0.21 \cdot 10^{-3}$   \\ \hline
\end{tabular}
\caption{Table of values of the quantum resonance $\en_1 =k_1^2$ and of the numerical coefficients $c_1$ and $a_1$ of formula (\ref {Formula20}) corresponding to the linear problem with repulsive singular potential for different values of the 
strength $\alpha$; for argument's sake sake we fix  the units such that $a=1$. \ The parameter $\Delta$ is the maximum of the absolute value of the difference between the survival amplitude computed with formula (\ref {Formula20}) and the survival amplitude computed with formula (\ref {Formula15}) and Theorem \ref {Teo1} for $t\in [0.5,5]$}
\label{Tavola2}
\end{center}
\end {table}

\begin{figure}
\begin{center}
\includegraphics[height=6cm,width=8cm]{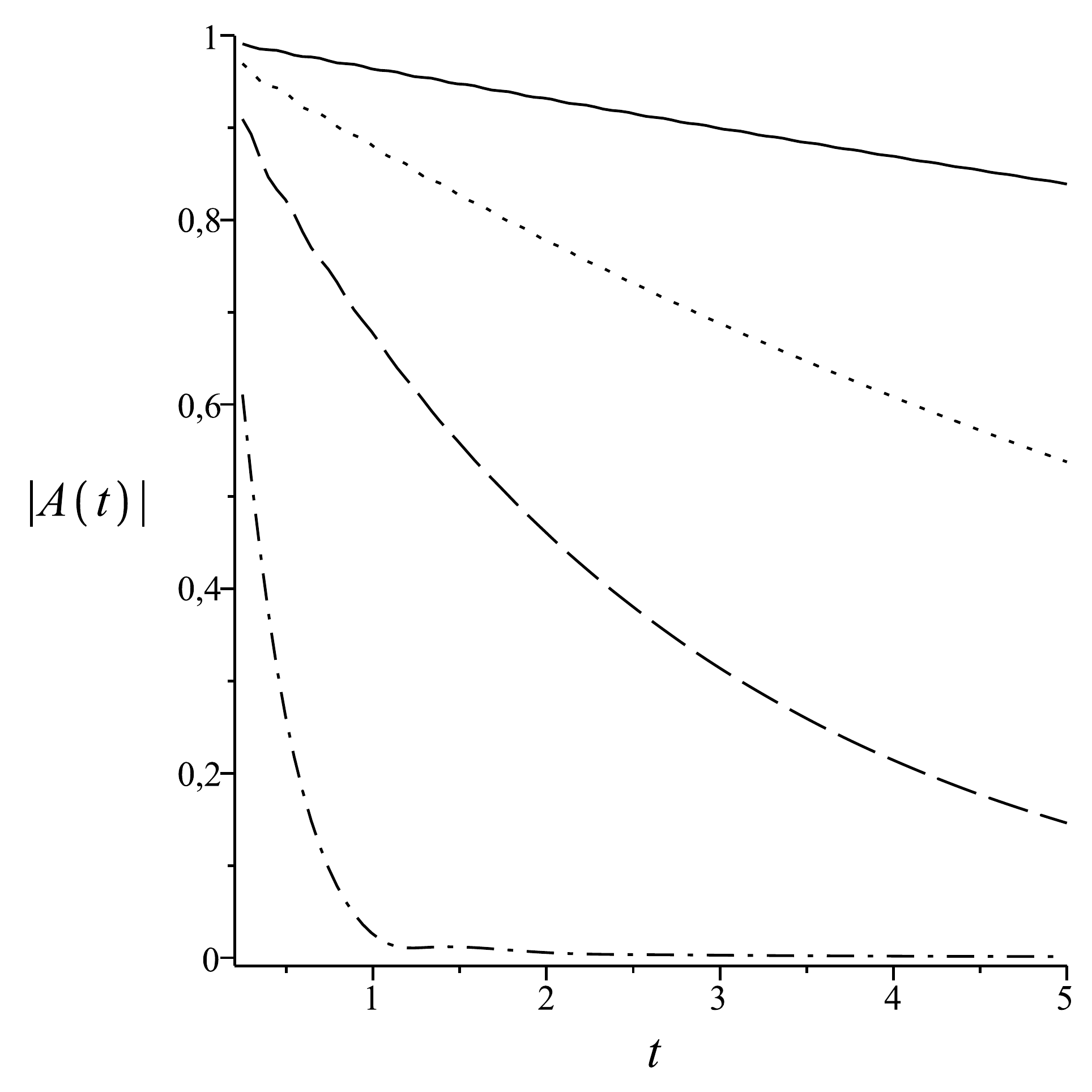}
\caption{Absolute value of the survival amplitude $ |\Am (t)  |$, where $\psi_0$ is given by $\psi_{\infty, 1}$,  for different values of $\alpha$ (full  line corresponds to $\alpha =40$, dot line corresponds to $\alpha =20$, broken line corresponds to $\alpha =10$ and broken-dot line corresponds to $\alpha =1$).}
\label {Fig_4}
\end{center}
\end{figure}

\section {Analysis of the nonlinear Winter's model} \label {S3}

In this section we consider stationary states and quantum resonances for the nonlinear Schr\"odinger equation 
\bee
\left \{
\begin {array}{l}
i\dot \psi_t = H \psi_t + \eta |\psi_t |^2 \psi_t  \\
\left. \psi_t \right |_{t=0} = \psi_0 
\end {array}
\right.  \, , \ \psi_t \in L^2 (\R^+) \, , \ \| \psi_0 \| = 1\, .  \label {Formula1}
\eee
As in the previous Section let us omit the dependence on $\alpha$ when this fact does not cause misunderstanding. \ In this Section we denote by $\Omega$ the energy value of the stationary states.

\begin {remark} \label {Nota8}
Similarly to the case of a single Dirac's $\delta$ potential \cite {ASa} one expects that the solution to (\ref {Formula1}) globally exists; however, we don't dwell here on a detailed proof of this result. \ Furthermore, a formal straightforward calculation proves the conservation of the norm and of the energy, i.e.
\be
\| \psi_t \| = \| \psi_0 \| \ \mbox { and } \ En (\psi_t )= En (\psi_0) \, , \ \forall t \, , 
\ee
where
\be
En (\psi )= \langle \psi , H \psi \rangle + \frac 12 \eta \| \psi \|_{L^4}^4 \, . 
\ee
\end {remark}

\subsection {Stationary states - preliminary results} \label {Ss31} We look for stationary solutions to the equation (\ref {Formula1}); that is $\psi_t (x) = e^{-i\En t} \psi (x)$ where ${\En}$ is real-valued and $\psi (x)$ is a solution to the equation
\bee
H \psi + \eta |\psi |^2 \psi = {\En} \psi  \,  , \psi \in L^2 (\R^+ ) \, ,\ \| \psi \| =1 \, . \label {Formula21}
\eee

\begin {remark} \label {Nota11}
We should point out that when one looks for stationary solutions to the linear problem (\ref {Formula2}) the normalization condition $\| \psi \| =1$ does not play a crucial role, we only have to require that $\psi \in L^2$. \ This is not the case in nonlinear Schr\"odinger equations; indeed, for any fixed $\eta$ the normalization condition $\| \psi \| =c$  affects the energy ${\En}$ of the associated stationary solutions. \ For argument's sake and without loosing in generality we fix such a value $c$ equal to one; if not we simply rescale
\be
\psi \to \frac {\psi }{c} \ \mbox { and } \ \eta \to c^2 \eta \, .
\ee
\end {remark}

First of all we prove that if a solution $\psi $ to (\ref {Formula21}) there exists then $\psi $ is, up to a constant phase factor, a real-valued function. 

\begin {proposition} \label {Prop3}
Let $\psi \in L^2 (\R^+ )$ be a solution to the nonlinear equation  (\ref {Formula21}), where ${\En}$ and $\eta$ are real-valued, satisfying conditions (\ref {Formula4}) and (\ref {Formula5}); then $\psi (x)$ is, up to a constant phase factor, a real-valued function.
\end {proposition}

\begin {proof}
We have that 
\be
W(x):= \left [ \psi' \bar \psi - \psi \bar \psi ' \right ] = 
\left \{
\begin {array}{ll}
c_1 & \mbox { if } x \in (0,a) \\ 
c_2 & \mbox { if } x \in (a,+\infty ) 
\end {array}
\right. 
\ee
is a piece-wise constant function. \ Indeed, since $\psi$ satisfies to the equation
\bee
-\psi '' + \eta |\psi |^2 \psi = {\En} \psi \, , \ x \in (0,a) \cup (a,+\infty ) \, , \label {Formula22}
\eee
then if we multiply both sides by $\bar \psi$ and   take the difference of the resulting terms with their complex conjugate we have that 
\bee
\frac {dW}{dx} = \left [ \psi'' \bar \psi - \psi \bar \psi '' \right ] = 0 \, , \ x \in (0,a) \cup (a,+ \infty ) \, , \label {Formula23}
\eee
since $\eta $ and $\En$ are real-valued parameters. \
Let 
\be
W_\pm = W(a\pm 0) :=\lim_{x\to a^\pm} W(x)
\ee
be the right (+) and left (-) limit of $W(x)$ at $x=a$; then conditions (\ref {Formula5}) imply that 
\be
W_+ &=& \psi' (a+0) \bar \psi (a+0) - \psi (a+0) \bar \psi ' (a+0)  \\
&=& \left [ \psi' (a-0) + \alpha \psi (a+0) \right ] \bar \psi (a+0) - \psi (a+0) \left [ \bar \psi' (a-0) + \alpha \bar \psi (a+0) \right ]  \\
&=&  \psi' (a-0) \bar \psi (a-0) - \psi (a-0) \bar \psi' (a-0) = W_-\, . 
\ee
Hence, $c_2=c_1$. \ Furthermore, condition (\ref {Formula4}) implies that $c_1=0$ since $\psi (0) =0$. \ Now, if we set
\be
\psi (x) = 
\left \{ 
\begin {array}{ll}
\phi_1 (x) e^{i\theta_1 (x)} & \mbox { if } x \in (0,a) \\
\phi_2 (x) e^{i\theta_2 (x)} & \mbox { if } x \in (a,+\infty)
\end {array}
\right.
\ee
where $\phi_{1,2}(x) \ge 0 $ and $\theta_{1,2} (x)$ are real-valued, then equation $\psi' \bar \psi - \psi \bar \psi ' =0$ implies that $\theta_j$ are constant functions and thus the matching conditions (\ref {Formula5}) implies that $\theta_2 - \theta_1 = 2n \pi $ for some integer number $n$. 
\end {proof}

\begin {remark} \label {Nota12}
From Proposition \ref {Prop3} it follows that when  $\eta$ and $\En$ are real-valued then equation (\ref {Formula21}) takes the form
\bee
H \psi + \eta \psi^3 = {\En} \psi  \, , \ \| \psi \| =1 \, . \label {Formula24}
\eee
We have to point out that this is not the case when $\En$ is complex-valued with non-zero imaginary part. 
\end {remark}

Now, we look for solutions to (\ref {Formula24}) in the case where $\alpha = + \infty$ at first and then for any $\alpha \in \R$.

\subsection {Stationary states - Infinite barrier: $\alpha = + \infty$.} \label {Ss32}

We separately treat the case of de-focusing nonlinearity, where $\eta >0$, and the case of focusing nonlinearity, where $\eta <0$.

\subsubsection {De-focusing nonlinearity: $\eta >0$.} \label {Ss321}

 It is well known \cite {D} that the general real-valued solution to the equation
\be
-\psi '' + \eta \psi^3 = {\En} \psi \, , \ \eta >0 \, , 
\ee
may be written as 
\be
\psi (x) = C \mbox {sn} \left ( \lambda (x-x_0 ) , p \right ) \, , \ p \in [0,1]\, , 
\ee
where $\mbox {sn} (x,p)$ is the Jacobi elliptic function and 
\be
p^2 = -\frac {\lambda^2-{\En}}{\lambda^2} \ \mbox { and } \ C^2 = - \frac {2(\lambda^2-{\En})}{\eta} = \frac {2p^2 \lambda^2}{\eta}\, ,
\ee
for some $C,\lambda \in \R$. \ In such a case the solution to (\ref {Formula24}) when $\alpha =+\infty$ is given by 
\be
\psi (x) = 
\left \{
\begin {array}{ll}
C \mbox { sn} \left ( \lambda (x-x_0 ) , p \right ) & ,\ x \in (0,a)  \\ 
0 & ,\ a< x  
\end {array}
\right. \, ,
\ee
with Dirichlet boundary conditions 
\bee
\psi (0)=\psi (a)=0 \, ; \label {Formula25}
\eee
that is
\be
\left \{
\begin {array}{l}
 \mbox { sn} \left ( \lambda x_0  , p \right ) = 0 \\
 \mbox { sn} \left ( \lambda (a-x_0 ) , p \right )=0 
 \end {array}
 \right. 
 \, .
\ee
Hence $ x_0 =0$ is a zero of the Jacobi elliptic function $\mbox {sn}$ and $\lambda$ is such that
\be
\lambda a = 2m {\KEll}(p) \, , \ m=1,2,\ldots \, ,
\ee
where ${\KEll}(p)$ is the value of the complete elliptic integral of first kind. 

The normalization condition implies that 
\be
1 = C^2 \int_0^a \left [ 
\mbox { sn} \left ( \lambda (x-x_0 ) , p \right ) \right ]^2 d x = C^2 \frac {2m}{\lambda} \frac {{\KEll}(p)-{\EEll}(p)}{p^2}
\ee
where ${\EEll}(p)$ is the complete elliptic integral  of second kind.

In conclusion, the following conditions must be satisfied
\be
\left \{
\begin {array}{l}
x_0 =0\\
\lambda a = 2m {\KEll}(p) \\
C^2 \frac {2m}{\lambda} \frac {{\KEll}(p)-{\EEll}(p)}{p^2} =1 \\
C^2 = \frac {2}{\eta} p^2 \lambda^2 
\end {array}
\right. 
\ee
that imply the following equation for $p \in [0,1]$:
\bee 
{\mathcal G}_+(p):={\KEll}(p) \left [ {\KEll}(p) - {\EEll}(p) \right ] =\frac {a\eta }{ 8m^2} \, . \label {Formula26}
\eee

Since the function ${\mathcal G}_+(p)$ is a monotone increasing function such that
\be
{\mathcal G}_+ (0+0)=0 \ \mbox { and } \ {\mathcal G}_+ (1-0) = + \infty
\ee
then equation (\ref {Formula26}) has exactly one real-valued solution $p_m \in [0,1)$ for any $m\in \N $.

In conclusion, we have proved that

\begin {proposition} \label {Prop4}
Let $\eta >0$ and let $p_m \in [0,1)$ be the unique solution to the equation
\be
{\KEll}(p) \left [ {\KEll}(p) - {\EEll}(p) \right ] =\frac {a\eta }{ 8m^2}
 \, , \ m\in \N \, . 
\ee
Let
\be
\lambda_m = \frac {2m{\KEll}(p_m)}{a}\, ,  \ C_m =\sqrt {2/\eta } p_m \lambda_m \ \mbox { and } \ {\En}_m = \lambda_m^2 (1+p_m^2) \, . 
\ee
Then 
\bee
\psi_m (x) = 
\left \{
\begin {array}{ll}
C_m \mbox {\rm sn} \left ( \lambda_m x , p_m \right ) & ,\ x \in (0,a)  \\ 
0 & ,\ a< x  
\end {array}
\right. \, , \label {Formula27}
\eee
is a stationary solution to (\ref {Formula24}) normalized to one.
\end {proposition}

\begin {remark} \label {Nota13}
In the limit case of $\eta \to 0$ then $p_m\to 0$ for any $n$, hence
\be
{\En}_m\to \lambda^2_m \to \left ( \frac {2m}{a}{\KEll}(0)\right )^2 = \left ( \frac {m\pi}{a} \right )^2  =  \en_{\infty , m} 
\ee
in agreement with the linear model.
\end {remark}

\subsubsection {Stationary states - Focusing nonlinearity: $\eta <0$.} \label {Ss322}

When $\eta <0$ then Proposition \ref {Prop4} takes the form

\begin {proposition} \label {Prop5}
Let $\eta <0$ and let $p_m \in [0,1)$ be the unique solution to the equation
\be
{\KEll}(p) \left [ {\EEll}(p) - (1-p^2) {\KEll}(p) \right ]  =  \frac {a|\eta |}{ 8m^2} \, , \ m\in \N . 
\ee
Let
\be
\lambda_m = \frac {2m{\KEll}(p_m)}{a} \, , \ C_m =\sqrt {2/|\eta |} p_m \lambda_m \ \mbox { and } \ {\En}_m = \lambda_m^2 (1-2p_m^2) \, . 
\ee
Then 
\bee
\psi_m (x) = 
\left \{
\begin {array}{ll}
C_m  \mbox {\rm cn} \left ( \lambda_m x-{\KEll}(p_m) , p_m \right )& ,\ x \in (0,a)  \\ 
0 & ,\ a< x  
\end {array}
\right. \, , \label {Formula28}
\eee
is a stationary solution to (\ref {Formula24}) normalized to one.
\end {proposition}

\begin {remark} \label {Nota14}
Like in the case of de-focusing nonlinearity even in this case we have that $\En_m \to \en_{\infty , m}$ as $\eta \to 0$. 
\end {remark}

\begin {proof} It is well known \cite {D} that the general real-valued solution to the equation
\be
-\psi '' + \eta \psi^3 = {\En} \psi \, , \ \eta < 0 \, , 
\ee
may be written in the form
\be
\psi (x) = C \mbox {cn} \left ( \lambda (x-x_0 ) , p \right ) 
\ee
where 
\be
p^2 = \frac {\lambda^2-{\En}}{2\lambda^2} \ \mbox { and } \ C^2 = - \frac {\lambda^2-{\En}}{\eta}= - \frac {2p^2 \lambda^2}{\eta} \, ,
\ee
for some $C,\lambda \in \R$. \ In such a case the solution to (\ref {Formula24}) when $\alpha =+\infty$ is given by 
\be
\psi (x) = 
\left \{
\begin {array}{ll}
C \mbox { cn} \left ( \lambda (x-x_0 ) , p \right ) & ,\ x \in (0,a)  \\ 
0 & ,\ a< x  
\end {array}
\right. \, ,
\ee
with Dirichlet boundary conditions (\ref {Formula25}); 
that is
\be
\left \{
\begin {array}{l}
 \mbox { cn} \left ( \lambda x_0  , p \right ) = 0 \\
 \mbox { cn} \left ( \lambda (a-x_0 ) , p \right )=0 
 \end {array}
 \right. 
 \, .
\ee
Hence $\lambda x_0 ={\KEll}(p)$ is a zero of the Jacobi elliptic function $\mbox {cn}$ and $\lambda$ is such that
\be
\lambda a = 2m {\KEll}(p) \, , \ m=1,2,\ldots \, .
\ee

The normalization condition implies that 
\be
1 &=& C^2 \int_0^a \left [ 
\mbox { cn} \left ( \lambda (x-x_0 ) , p \right ) \right ]^2 d x = m\frac {2C^2}{\lambda p^2}  \left [ {\EEll}(p) - (1-p^2) {\KEll}(p) \right ] \, . 
\ee

In conclusion, the following conditions must be satisfied
\be
\left \{
\begin {array}{l}
\lambda x_0 ={\KEll}(p)\\
\lambda a = 2m {\KEll}(p) \\
2m\frac {C^2}{\lambda p^2}  \left [ {\EEll}(p) - (1-p^2) {\KEll}(p) \right ] =1 \\
C^2 = - \frac {2}{\eta} p^2 \lambda^2 
\end {array}
\right. 
\ee
that imply the following equation for $p \in [0,1]$
\bee 
{\mathcal G}_- (p):={\KEll}(p) \left [ {\EEll}(p) - (1-p^2) {\KEll}(p) \right ] =\frac {a|\eta |}{ 8m^2} \, . \label {Formula29}
\eee

Since the function ${\mathcal G}_- (p)$ is a monotone increasing function such that
\be
{\mathcal G}_- (0+0)=0 \ \mbox { and } \ {\mathcal G}_- (1-0) = + \infty
\ee
then equation (\ref {Formula29}) has exactly one real-valued solution $p_m \in [0,1)$ for any $m \in \N  $ and any $\eta < 0$. \ Proposition \ref {Prop5} is so proved.
\end {proof}

\subsection {Stationary states - Finite barrier: $\alpha \in \R$.} \label {Ss33}

Recalling that 
\be
\lim_{p \to 1^-} \mbox {sn}(u,p) = \mbox {tanh} (u) \ \mbox { and } \lim_{p \to 1^-} \mbox {cn}(u,p) = \mbox {sech} (u)
\ee
and since we look for a real-valued solution $\psi (x)$  to (\ref {Formula24}) such that $\psi (x) \to 0$ as $x\to + \infty$ then such a solution there exists only when $\eta <0$ and ${\En}<0$ and it has the form 
\be
\psi (x) = 
\left \{
\begin {array}{ll}
C \mbox { cn} \left ( \lambda (x-x_0 ) , p \right ) & ,\ x \in (0,a) \ \mbox { where } p^2 = \frac {\lambda^2 - \Omega}{2\lambda^2} \ \mbox { and } \ C^2 = - \frac {\lambda^2 - \Omega}{\eta} \\ 
C' \mbox { sech} \left ( \lambda' (x-x_0' ) \right ) & ,\ a< x \ \mbox { where } {\lambda'}^2 =-{\En} \ \mbox { and } \ {C'}^2 =  \frac {2{\En}}{\eta} 
\end {array}
\right. \, .
\ee
Hereafter,  we may assume, for argument's sake, that $\lambda >0$ and $\lambda ' >0$.

The Dirichlet boundary condition $\psi (0)=0$ at $x=0$ implies that $\lambda x_0 ={\KEll}(p)$ is a zero of the Jacobi elliptic function $\mbox {cn}$. \ The matching condition (\ref {Formula5}) at $x=a$ implies that 
\be
\left \{
\begin {array}{l}
C' \mbox {sech} \left ( \lambda' (a-x_0' ) \right ) - C \mbox { cn} \left ( \lambda (a-x_0 ) , p \right ) =0 \\
C' \mbox {sech} (\lambda' (a-x_0')) \left [ \lambda' \mbox {tanh} (\lambda' (a-x_0')) + \alpha \right ] - C \lambda \mbox {sn}(\lambda (a-x_0),p) \mbox {dn}(\lambda (a-x_0),p)  = 0
\end {array}
\right.
\ee
Furthermore, we have to require the normalization condition:
\be
C^2 \int_0^a \left [ \mbox { cn} \left ( \lambda (x-x_0 ) , p \right ) \right ]^2 dx + {C'}^2 \int_a^{+\infty} \left [ \mbox { sech} \left ( \lambda' (x-x_0' ) \right ) \right ]^2 dx =1 
\ee
that is 
\be
\frac {C^2}{\lambda} G(\lambda a)  + \frac {{C'}^2}{\lambda'} \left [1- \mbox {tanh} \left (  \lambda' (a- x_0') \right )  \right ]  =1 
\ee
where 
\be
G(u):= \int_0^{u} \left [ \mbox { cn} \left ( q + {\KEll}(p), p \right ) \right ]^2 dq = \frac {u-{\EEll} \left [ \mbox {sn}(u,p)\right ] }{p^2} + 2 \tilde n \frac {{\KEll}(p)-{\EEll}(p)}{p^2} 
\ee
where $ \tilde n = \mbox {round} \left [ u/2/{\KEll} (p) \right ]$. \ In conclusion, the parameters $C$, $C'$, $\lambda$, $\lambda'$, $ x_0$, $ x_0'$  and $ p$ must satisfy to the following conditions
\be
\left \{
\begin {array}{l}
\lambda x_0 ={\KEll}(p)\\
C' \mbox {sech} \left ( \lambda' (a-x_0' ) \right ) - C \mbox { cn} \left ( \lambda (a-x_0) , p \right ) =0 \\
C' \mbox {sech} (\lambda' (a-x_0')) \left [ \lambda' \mbox {tanh} (\lambda' (a-x_0')) + \alpha \right ] - C \lambda \mbox {sn}(\lambda (a-x_0),p) \mbox {dn}(\lambda (a-x_0),p)  = 0 \\
\frac {C^2}{\lambda} G (\lambda a)  + \frac {{C'}^2}{\lambda'} \left [1- \mbox {tanh} \left (  \lambda' ( a- x_0') \right )  \right ]  =1 \\
C^2 = - \frac {2}{\eta} p^2 \lambda^2 \\
 {C'}^2 = - \frac {2{\lambda'}^2}{\eta} \\
 (1-2p^2) \lambda^2 = -{\lambda '}^2 
\end {array}
\right. 
\ee
with some constrains: e.g. $\eta <0$ and $2p^2 -1 \ge 0$, that is $p \in [1/\sqrt {2},1]$.

The numerical study of such a system of equations gives that (see Figure \ref {Fig_5})

\begin {proposition} \label {Prop6}
For any $\alpha \not=0 $ there exists $\tilde \eta (\alpha )<0$ such that if $\eta \le  \tilde \eta (\alpha )$ then stationary solutions corresponding to some negative real values $\En_n^\pm (\eta )$, $n=1,2,\ldots , N(\eta )$ for some positive integer $N(\eta )$, there exists; these couple of stationary solutions comes from a saddle point bifurcation occurring at $\eta = \eta_n (\alpha )$ for some $\eta_n (\alpha )$, where $\eta_{n+1} (\alpha ) \le \eta_n (\alpha )$ and $\eta_1 (\alpha )=\tilde \eta (\alpha )$. \ Furthermore, when $\alpha <0$ is such that $a \alpha <-1$ then there exists a stationary solution corresponding to $\En_0 (\eta )$ for any $\eta <0$ such that
\be
\lim_{\eta \to 0-0} \En_0 (\eta )= - \left [ \frac {1}{2a} [-a \alpha + W_0 (a\alpha e^{a \alpha } ) \right ]^2 \, . 
\ee
\end {proposition}

\begin {remark}\label {NotaNuova1}
For a given valu of $\eta$, e.g. $\eta = -7.4$ then two stationary states there exist with energy $\En_1^+ =-0.36$ and $\En^-_1 = - 2.29$; the corresponding values of $p$,  $\lambda$, $C$ and $x_0$ are $p_1^+=0.72$, $p_1^- =0.77$, $\lambda_1^+ = 3.39$, $\lambda_- =3.56$, $C^+=1.27$, $C^- =1.43$, $x_0^+ =0.55$ and $x_0^- =0.55$. \ If we denote
\be
I = \int_0^a |\psi (x) |^2 dx = C^2 \int_0^a \left | \cn \left ( \lambda (x-x_0),p \right ) \right |^2 dx 
\ee
then
\be
I^+ = 0.80 \ \mbox { and } \ I^- =0.98 \, .
\ee
\end {remark}

\begin{figure}
\begin{center}
\includegraphics[height=5cm,width=5cm]{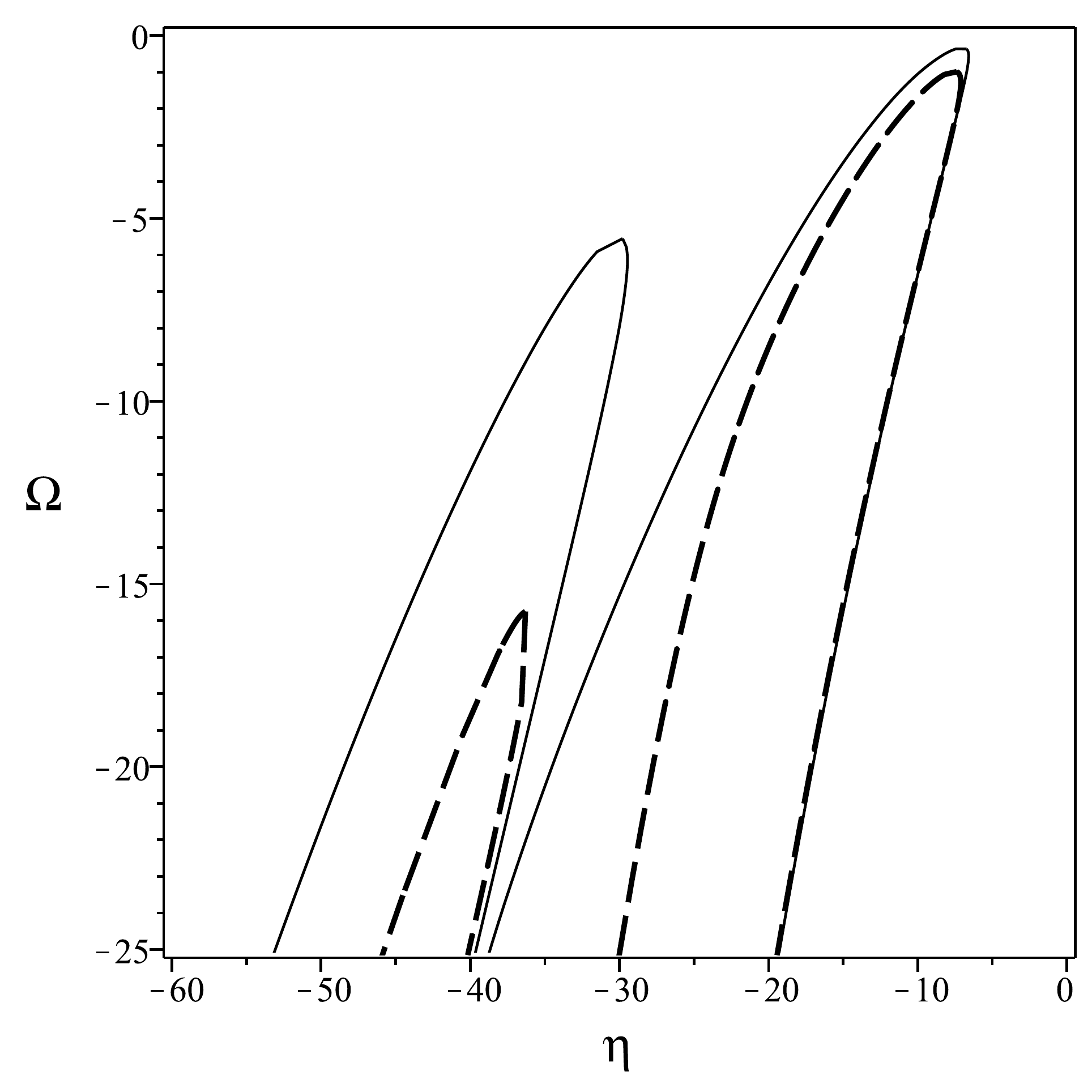}
\includegraphics[height=5cm,width=5cm]{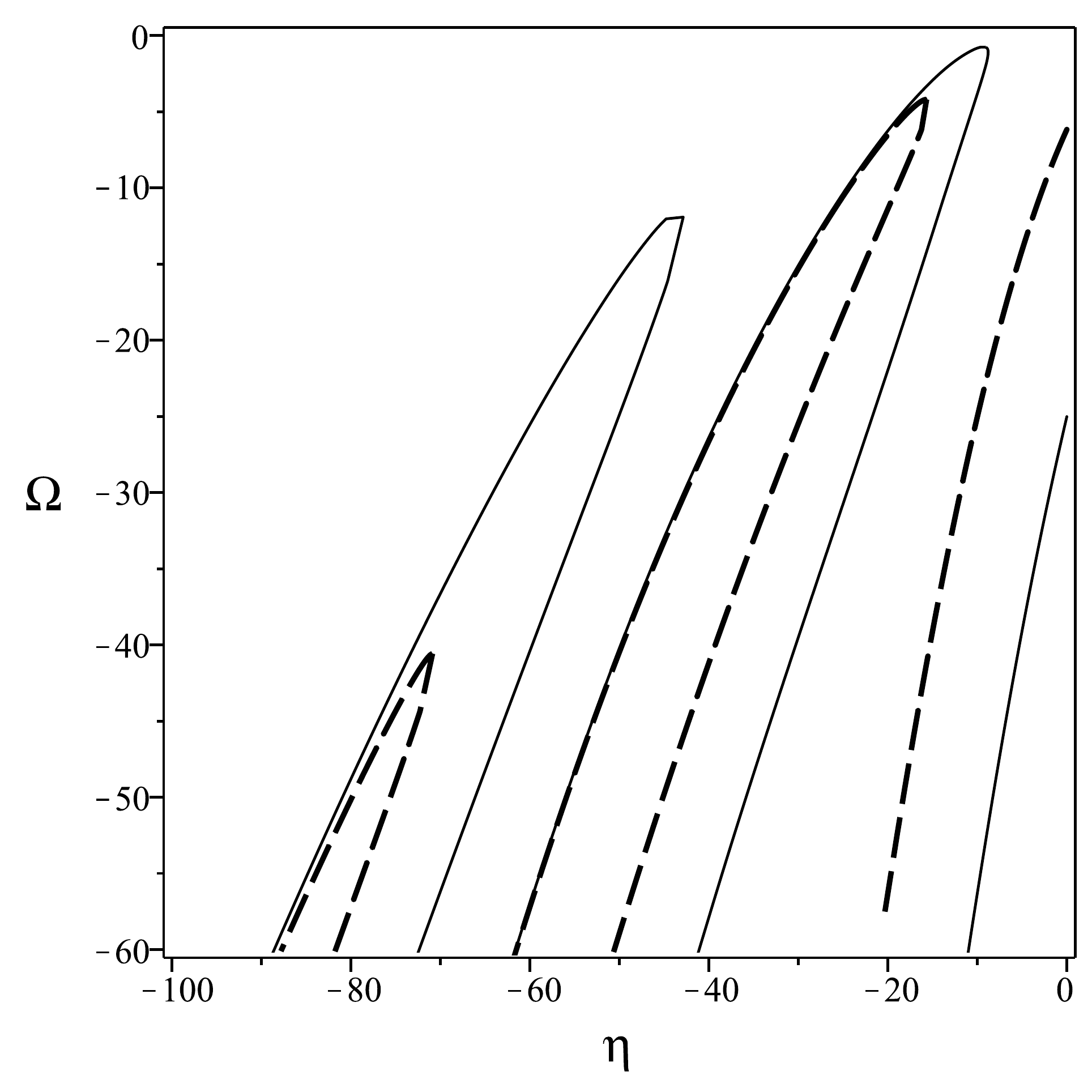}
\caption{Here we plot the energy value $\Omega$ of stationary solutions to (\ref {Formula24}) for different values of $\alpha$. \ In the left hand side picture we plot $\En_1^\pm (\eta )$, $\En_2^\pm (\eta )$ where full lines correspond to $\alpha =10$ and broken lines correspond to $\alpha=5$. \ In the right hand side picture we plot the $\En_1^\pm (\eta )$, $\En_2^\pm (\eta )$ and $\En_0 (\eta )$ where full lines correspond to $\alpha =-10$, broken lines correspond to $\alpha=-5$. \ $\En_{1,2}^\pm (\eta )$ are associated to saddle point bifurcations at $\eta =\eta_1 (\alpha )$ and $\eta = \eta_2 (\alpha )$, where $\eta_1 (\alpha )=\tilde \eta (\alpha )$ and where, for instance, $\tilde \eta (10) = -6.59$. \  $\En_0 (\eta )$ is such that $\En_0 (0-0)$ coincides with the value of the energy $\en$ given by (\ref {Formula10}).}
\label {Fig_5}
\end{center}
\end{figure}

\subsection {Quantum resonances for NLS} \label {Ss34}

Here, we critical review some definitions  proposed for quantum resonances in NLS. \ Only in this Section we denote by $\Gamma$ the strength of the nonlinear term instead of $\eta$; that is we consider the equation
\bee
H \psi + \Gamma \psi^3 = {\En} \psi  \label {Formula24Bis}
\eee
instead of (\ref {Formula24}). \ The reason of this choice will be explained in Remark \ref {R18}.

\subsubsection {On the definition of quantum resonances for NLS by the complex scaling method} \label {Ss341} This technique basically consists to the application of the  mapping $ \psi (x) \to \psi_\theta (x) = e^{i\theta /2} \psi (x e^{ i\theta } )$, $\theta \in \C$. \ This map is not unitary and applying
this transformation to the linear Schr\"odinger operator the bound states, if there, are $\theta$-independent while the continuum spectrum is “rotated” in the complex energy plane \cite {CFKS}. \ Such a rotation uncovers quantum resonances, which correspond to the poles of the analytical continuation of the kernel of the resolvent operator as discussed in \S \ref {Ss22}. \ Some attempts to extend such a definition of quantum resonances for NLS have been recently proposed \cite {MCMB,MC,SP,SW} but a serious drawback occurs: the nonlinear term $|\psi (x)|^2$ is not an analytic function! \ This fact has been properly recognized from the authors of the  papers quoted above and some expedients, e.g. substituting $|\psi |^2$ by $\psi^2$, have been proposed in order to circumnavigate this problem, but we think that the fact that analyticity property is lost cannot be fixed and furthermore the substitution of $|\psi |^2$ by $\psi^2$ changes the nonlinear Winter's model.  

\subsubsection {On the definition of quantum resonances for NLS by the Siegert's approximation method} \label {Ss342} The Siegert's approximation method, explained in Remark \ref {Nota3}, yields good results for narrow resonances for linear Schr\"odinger operators \cite {R}. \ Some authors \cite {RK} proposed to apply such a method to the study of resonances for NLS, too. \ In particular, in the case of Winter's nonlinear model they made the ansatz that the resonance wavefunction is such that $\psi (x) = C e^{ik x}$, for $x>a$, where $k$ is related to the resonance energy by the formula (in our notation) $\Omega = k^2+\eta |C|^2$. \ We should point out that this ansatz property works only for real-valued $\Omega$; unfortunately, when the imaginary part of $\Omega$ is not exactly zero then $\psi (x) = C e^{ik x}$, for $x>a$, is not a solution to the Winter's nonlinear model. \ For a for single well/barrier NLS some authors \cite {M,RWK} circumnavigate this problem modifying equation (\ref {Formula24Bis}) neglecting the nonlinear term outside the potential well in order to unambiguously define ingoing and outgoing waves and thus a scattering coefficient. \ They motivate this approach by pointing out that such an  approximation is well justified for bound states since in this case the condensate density is much higher inside the potential well than outside. \ We think that this artifact enable us to properly define the outgoing condition but it does not solve a second fatal problem; that is explicit solutions to NLS is only known when $\En$ is real-valued and where the nonlinear term is given by $\eta \psi^3$ and not by $\eta |\psi |^2 \psi$. 

\subsubsection {On the definition of quantum resonances for NLS by the analysis of scattering coefficients} \label {Ss343}  Let us define, as done in the linear case, the \emph {scattering coefficient} simply as
\be
S = \frac {C^2}{{C'}^2}
\ee
where $C$ and $C'$ are the coefficients of the total wave function in regions $0<x<a$ and $a<x$ \cite {CMBS,WMK}. \ Resonances for the Winter's nonlinear model may be defined as the relative maximum value points of $S$. \ Here, following the treatment adopted by \cite {WMK}, we compute the scattering coefficient $S(\En )$ as function of the energy $\En$ and then we plot its graph. \ To this end we have to look for solutions to (\ref {Formula24Bis})  satisfying the boundary conditions (\ref {Formula4}) and (\ref {Formula5}), furthermore we assume some ``normalization'' condition; since $\psi \notin L^2 (\R^+ )$ then we cannot consider the usual normalization condition $\| \psi \|=1$ and then we choose to put the normalization condition on the well, that is
 \be
 \| \psi \|_{L^2 ([0,a])} =1 \, . 
 \ee

\begin {remark} \label {Nota15}
We should point out that the expression of the scattering coefficient $S(\En )$ for real-valued $\En$ is obtained for real-valued $\psi$ where the nonlinear term  $\eta |\psi |^2 \psi$  in equation (\ref {Formula24Bis})is substituted by $\eta \psi^3$ and then we cannot extend the resulting expression of $S(\En )$ to complex-valued $\En$. 
\end {remark}

\begin{table}
\begin{center}
\begin{tabular}{|c|c|c|c||c|c|c|c|} \hline
\multicolumn {4}{|c||}{$\Gamma =-9$}& \multicolumn {4}{|c|}{$\Gamma =+9$} \\
\hline
$\alpha $  &  $\En_1 $   & $ \En_2 $   &  $\En_3 $ & $\alpha $  &  $\En_1 $   & $ \En_2 $   &  $\En_3 $ \\   
 \hline
$1 $  &  $ 2.95 $   & $ 26.35 $  & $70.6$ & $ 1$  &  $17.64$  & $43.0$  & $88.0$  \\ \hline
$5 $  &  $ 3.73 $   & $ 28.3 $  & $72.4 $ & $ 5$  &  $19.18$  & $45.4$  & $89.8$  \\ \hline
$10 $  &  $ 4.51 $   & $ 29.86 $  & $74.8$ & $ 10$  &  $20.21$  & $47.2$  & $92.2$  \\ \hline
\end{tabular}
\caption{The scattering coefficient $S(\En )$ assumes relative maximum values at the points $\En_n$, $n \ge 1$. \ In the table we collect these values for $n=1,2$ and $3$, for some values of $\alpha$ and in both focusing and de-focusing nonlinearity cases.}
\label{Tavola3}
\end{center}
\end {table}

Now, we are ready to compute the scattering coefficient; for argument's sake we restrict ourselves to the focusing case where $\Gamma <0$; the de-focusing case where $\Gamma >0$ can be similarly treated and thus we don't dwell here on the details.

\begin{figure}
\begin{center}
\includegraphics[height=5cm,width=5cm]{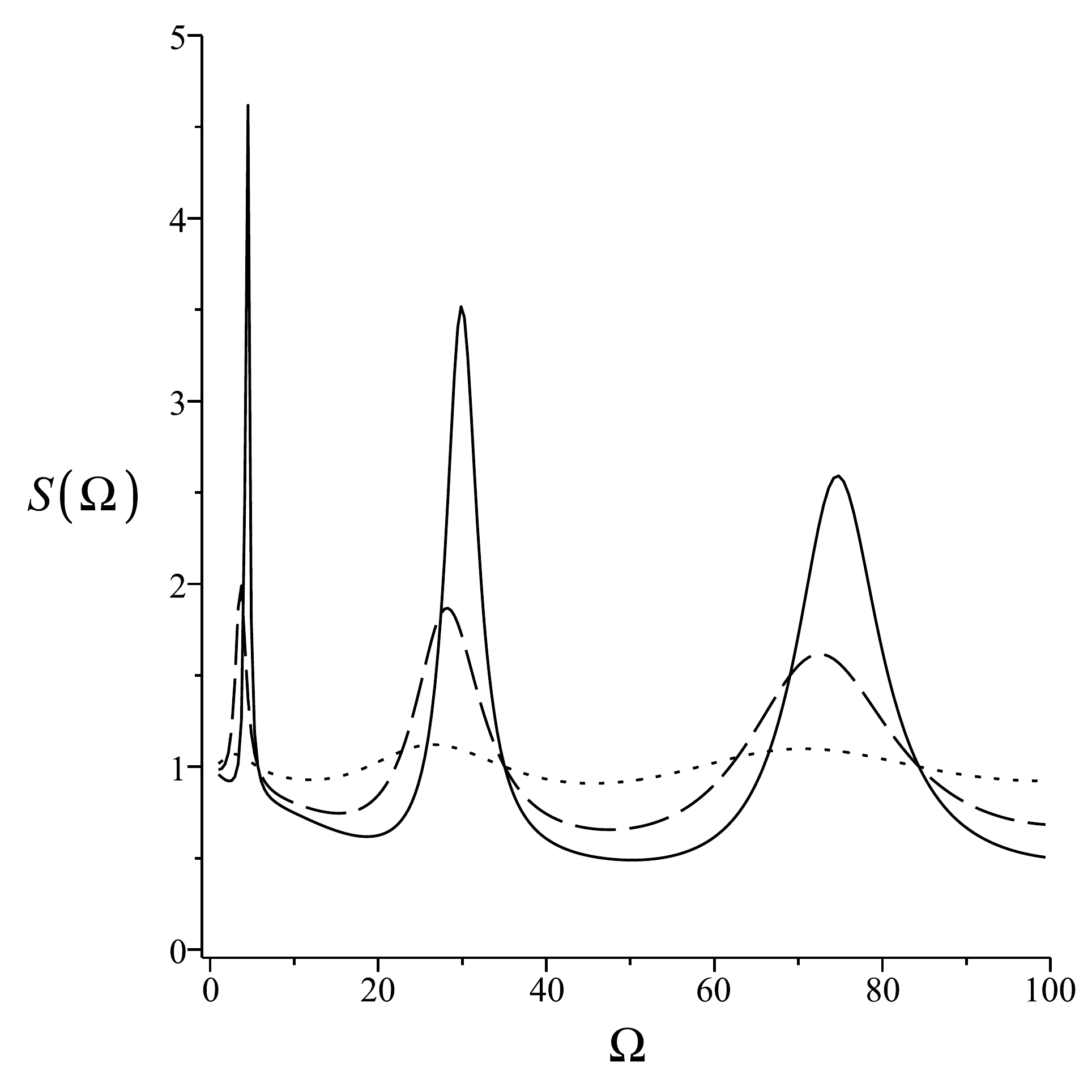}
\includegraphics[height=5cm,width=5cm]{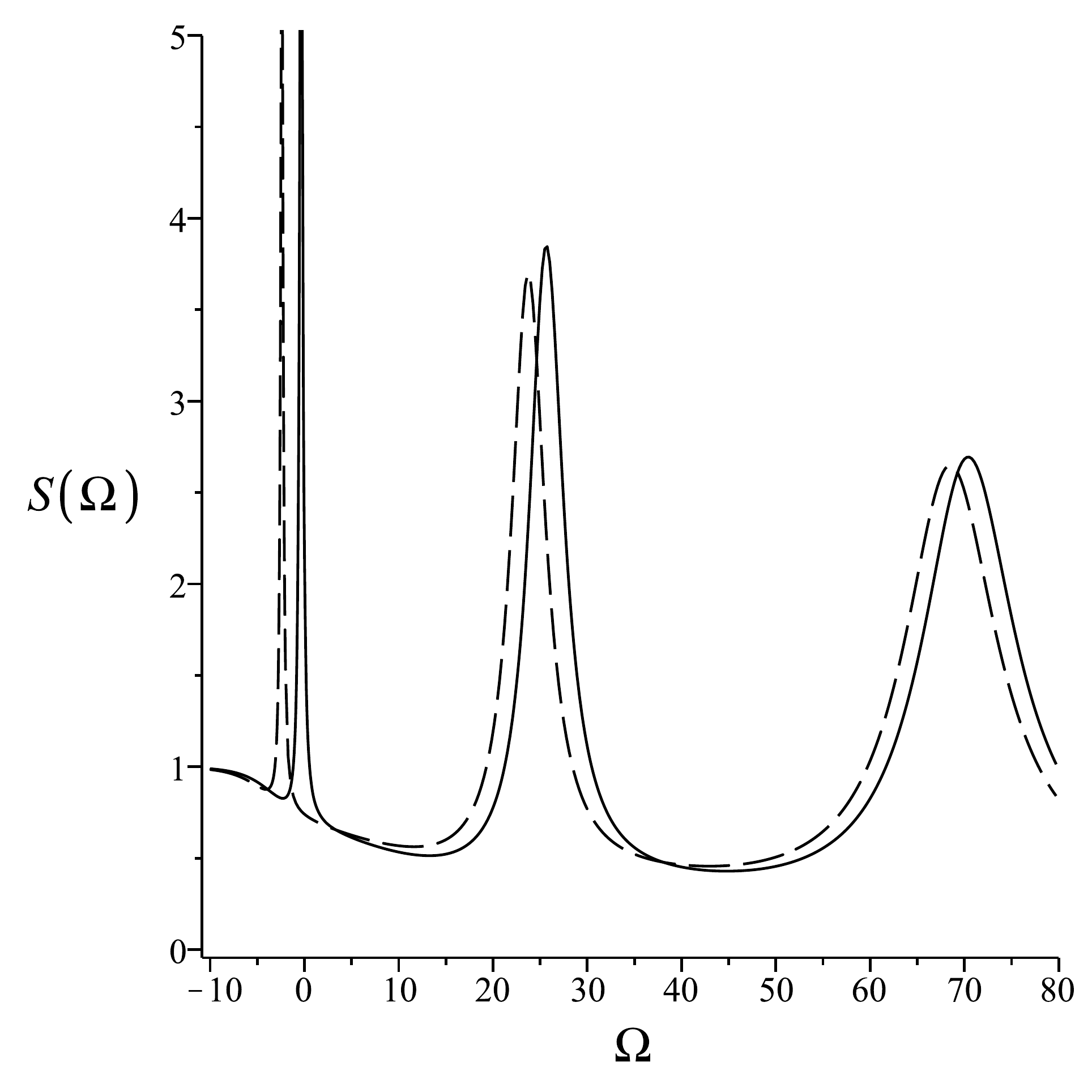}
\caption{ In the left hand side picture we plot the scattering coefficient $S(\En )$ for different values of $\alpha$ and $\Gamma =-9$: $\alpha =1$ (dot line), $\alpha =5$ (broken line) and $\alpha =10 $ (full line). \ In the right hand side picture  we plot the scattering coefficient $S(\En )$ for $\Gamma =-5.95$ (full line) and $\Gamma =-7.27$ (broken line); in both cases $\alpha =+10$.}
\label {Fig_6}
\end{center}
\end{figure}

Let $\Gamma <0$, let ${\En}$ be fixed and let
\be
\psi (x) = 
\left \{
\begin {array}{ll}
C \mbox { cn} \left ( \lambda (x-x_0 ) , p \right ) & ,\ x \in (0,a) \ \mbox { where } p^2 = \frac {\lambda^2-{\En}}{2\lambda^2}  \ \mbox { and } \ C^2 = - \frac {\lambda^2-{\En}}{\Gamma}  \\ 
C' \mbox { cn} \left ( \lambda (x-x_0' ) , p' \right ) & ,\ x >a \ \mbox { where } {p'}^2 = \frac {{\lambda'}^2-{\En}}{2{\lambda'}^2}  \ \mbox { and } \ {C'}^2 = - \frac {{\lambda'}^2-{\En}}{\Gamma} 
\end {array}
\right. \, .
\ee
be a real-valued solution to the nonlinear equation in (\ref {Formula24Bis}) ($\lambda >0$ and $\lambda ' >0$). \ If $C' \not= 0$ or $p' \not= 1$ then $\psi \notin L^2$. \ The Dirichlet boundary condition (\ref {Formula4}) at $x=0$ and the  matching condition (\ref {Formula5}) at $x=a$ imply that $\psi (0)=0$ and 
\be
\left \{
\begin {array}{l}
C' \mbox { cn} \left ( \lambda' (a-x_0' ) , p' \right ) - C \mbox { cn} \left ( \lambda (a-x_0 ) , p \right ) =0 \\
-C' \lambda' \mbox {sn}(\lambda' (a-x_0'),p') \mbox {dn}(\lambda' (a-x_0'),p')+ C \lambda \mbox {sn}(\lambda (a-x_0),p) \mbox {dn}(\lambda (a-x_0),p)  = \alpha C \mbox { cn} \left ( \lambda (a-x_0 ) , p \right )
\end {array}
\right.
\ee
In conclusion, the parameter $C$, $C'$, $\lambda$, $\lambda'$, $x_0$, $x_0'$, $p$ and $p'$ must satisfy to following conditions
\be
\left \{
\begin {array}{l}
\lambda x_0 ={\mathcal K}(p) \\
C' \mbox { cn} \left ( \lambda' (a-x_0' ) , p' \right ) - C \mbox { cn} \left ( \lambda (a-x_0 ) , p \right ) =0 \\
-C' \lambda' \mbox {sn}(\lambda' (a-x_0'),p') \mbox {dn}(\lambda' (a-x_0'),p')+ C \lambda \mbox {sn}(\lambda (a-x_0),p) \mbox {dn}(\lambda (a-x_0),p)  = \alpha C \mbox { cn} \left ( \lambda (a-x_0 ) , p \right ) \\
1
=  -\frac {2p^2 \lambda}{\Gamma } \int_0^{\lambda a} \mbox { cn}^2 \left ( y - {\mathcal K}(p), p \right ) dy\\
C^2 =  -\frac {2}{\Gamma} p^2 \lambda^2 \\
 {C'}^2 =  -\frac {2{\lambda'}^2{p'}^2}{\Gamma} \\
 {\En}=(1-2p^2) \lambda^2 = (1-2{p'}^2) {\lambda'}^2
\end {array}
\right.
\ee
From these equations one can plot in Figure \ref {Fig_6} the scattering coefficient $S (\En )$.  \ In Table \ref {Tavola3} the value of the energy $\Omega$ corresponding to the first three "resonances", identified with the maximum point of the scattering coefficient, is reported for, e.g., $\Gamma =- 9$ and $\Gamma =+9$.

\begin {remark} \label {Nota16}
Figure \ref {Fig_6} may suggest that, as in the linear case, complex poles for $S(\Omega )$ there exists and thus a typical exponentially decay time-behavior associated to quantum resonances would occur.
\end {remark}

\begin {remark} \label {R18}
From Figure \ref {Fig_6} it turns out that when the nonlinearity strength $\Gamma$ changes then the resonances energies shift and, in some cases, become narrow and narrow. \ In particular we expect that they eventually become the stationary states $\Omega_n^\pm$ obtained in Proposition \ref {Prop6} when $\Gamma$ takes some values. In order to compare the values of $\eta$ and $\Gamma$ for which stationary states occur we have to point out that we are dealing with two different normalization conditions; if we denote by $\phi (x)$ the stationary solution to equation (\ref {Formula24}) corresponding to a given $\eta <0$ and with the normalization condition 
\be
\int_0^{+\infty} | \phi (x) |^2 dx = 1 
\ee
then
\be
\psi (x) = \frac {1}{\sqrt {I}} \phi (x) \, ,
\ee
is a stationary solution to (\ref {Formula24Bis}) for $\Gamma = \eta I$, where $I$ is defined in Remark \ref {NotaNuova1}. \ Then, from Remark \ref {NotaNuova1} at $\Gamma = -7.27$, corresponding to $\eta =-7.4$, we expect to see a stationary state with energy $\En =\En_1^+ = -0.36$, and at $\Gamma =5.95$, still corresponding to the same value $\eta =-7.4$, we expect to see a stationary state with energy $\En =\En_1^- =-2.29$ (see the right hand side of  Figure \ref {Fig_6}).
\end {remark}

\subsection {Survival amplitude for NLS} \label {Ss35}

Let $\psi_t (x)$ be the solution to (\ref {Formula1}) with initial condition $\psi_0 (x)$; the  \emph {survival amplitude} is the scalar product between these two vectors, that is
\bee
\Am (t) := \langle \psi_0 , \psi_t \rangle \, . \label {Formula30}
\eee
In order to numerically compute $\psi_t$ we make use  of the spectral splitting method discussed below. \ In particular, we perform the following numerical experiments where $\psi_0 $ is given by the wavefunction to the linear Schr\"odinger operator $H_\infty$, i.e. $\psi_0 (x)= \psi_{\infty , 1} (x)$.

\subsubsection {Spectral splitting method} \label {Ss351} In order to compute the solution $\psi_t$ here we make use of the spectral splitting approximation method. \ The basic idea is
quite simple (see, e.g., the paper \cite {B} and, in particular, the paper \cite {Sacchetti1} where the spectral splitting method has been adapted to the case of singular potentials): suppose to consider a formal evolution equation 
\be
i \dot \psi_t =[A+B (\psi_t) ] \psi_t \, ,\psi_0 = \left. \psi_t \right |_{t=t_0}\, ,
\ee
where $A$ and $B(\psi_t )$ are two given operators, where the second one is not linear and it depends on $\psi_t$. \  Let us denote by $S^{t - t_0} \psi_0$ its solution where $S^{t - t_0}$ is the associated evolution
operator; let us denote by $X^{t - t_0}$ and $Y^{t - t_0}$ the evolution operators respectively associated to the equations
\be
i \dot \psi_t =A \psi_t \ \mbox { and } \ i \dot \psi_t =B(\psi_t) \psi_t \, . 
\ee
It is well known that, in general, $S^\delta \psi_0 - X^\delta Y^\delta \psi_0 \not= 0$, for any $\delta \in \R$, 
but this difference may be proved, under some circumstances, to be small when $\delta $ is small. \ That is for any fixed $t>0$ and any $n \in \N$ large enough we have that
\bee
\psi_t = S^{t-t_0} \psi_0 \sim (X^\delta Y^\delta)^n \psi_0 \, , \ \mbox { where } \ \delta =(t-t_0)/n \, . \label {Formula31}
\eee
We apply now such an approximation method to the NLS  (\ref {Formula1}) where $A$ is the linear operator $H$ and where $B (\psi_t )$ is the nonlinearity term $\eta |\psi_t |^2$. \ Hence $X^\delta = e^{-i H \delta }$ is the evolution operator discussed in \S \ref {Ss23}; concerning $Y^\delta$ we have to solve equation
\bee
i \dot w_t = \eta |w_t |^2 w_t \, , \ \mbox { where } \ w_{t_0}=w_0 \, ,  \label {Formula32}
\eee
and a straightforward calculation gives that $|w_t|^2$ is constant with respect to $t$; thus (\ref {Formula32}) can be written
\be
i \dot w_t = B(w_t) w_t = B({w_{t_0}}) w_t 
\ee
and then $Y^\delta$ is nothing but the simple multiplication operator
\be
Y^\delta w_0 = e^{-i \eta |w_0 |^2 \delta } w_0\, .
\ee
In conclusion, by means the spectral splitting method we have that (let us fix $t_0=0$ for argument's sake) 
\bee
\Am (t) = \langle \psi_0 , \psi_t \rangle  = \langle \psi_0 , S^t \psi_0 \rangle  \sim \langle \psi_0 , (X^\delta Y^\delta )^n \psi_0 \rangle \label {Formula33}
\eee
for $n$ large enough where $\delta =t/n$ and where 
\be
X^\delta \psi = e^{-iH\delta }\psi \ \mbox { and } \ Y^\delta \psi = e^{-i\eta |\psi |^2}\psi \, .
\ee

\subsubsection {Numerical experiments} In this Section we compute the survival amplitude (\ref {Formula30}) where $\psi_0$ is given by (\ref {Formula11}); for argument's sake we fix $a=1$ and $\alpha =+10$ and we compare five different numerical experiments where $\eta=0$ (i.e. the linear evolution), $\eta = \pm 3$ and $\eta =\pm 9$. \ In Figure \ref {Fig_7}, left-hand side picture, we plot the absolute value of the survival amplitude $| \Am (t) |$ for $t$ in a given interval (for argument's sake let $t \in [0,1]$).
\begin{figure}
\begin{center}
\includegraphics[height=5cm,width=8cm]{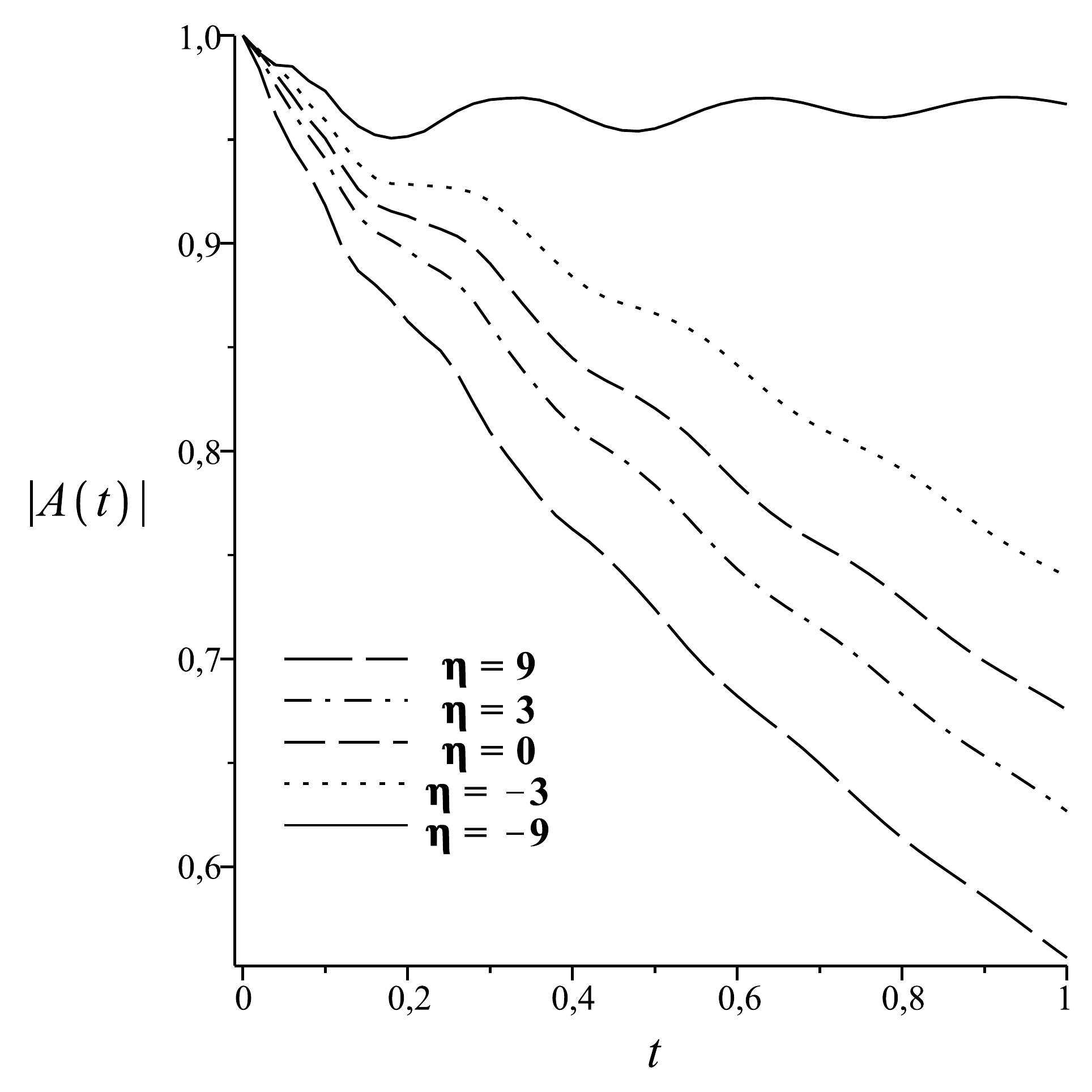}
\caption{Here we plot the absolutely value of the survival amplitude $|\Am (t)|$ where the initial wavefunctions $\psi_0$ is given by (\ref {Formula11}). \ In particular, $a=1$ and $\alpha =+10$ are fixed.}
\label {Fig_7}
\end{center}
\end{figure}

\begin {remark} \label {Nota17}
From Figure \ref {Fig_7} it turns out that a rather fast exponential    decay occurs when $\eta \ge 0$, while for $\eta <0$, such a time-decay effect seems to disappear. \ This fact may be explained by the fact that the quantum resonances associated to maximum point of the scattering coefficient $S(\Omega )$ have energies that increase when $\eta >0$ increases; in contrast, when $\eta <0$ then the resonance energy decreases and the resonances become narrow and narrow and eventually they become stationary state for $\eta \le \tilde \eta$.
\end {remark}

\section {Do quantum resonances have significance in nonlinear Schr\"odinger equations? Comments and conclusions} \label {Sec4}

A key feature of quantum resonances for the study of the exponentially decay of the survival amplitude for linear Schr\"odinger equations is that the exponent in formula (\ref {Formula19}) does not depend on the initial state but only on the imaginary part of the complex-valued energy of the quantum resonance. \ We'll see that this property still seems to be valid for NLS by means of the following numerical experiment in which we consider different initial conditions $\psi_{i,0}(x)$, $i=A,G,S$, and compare the exponential decays in the absence and presence of nonlinear terms in Figure \ref {Fig_8}.

\subsection {Numerical experiment}

Consider the case in which the following different initial conditions are assigned:

\begin {itemize}

\item [-] the initial wavefunction is the eigenfunction  (\ref {Formula11}) in the case where we put the Dirichlet condition in $x=a$ in the linear case:
\be
\psi_{A,0} (x) = \sqrt {\frac {2}{a}} \sin \left ( \frac {\pi x}{a} \right ) \chi_{(0,a)} (x) \, . 
\ee

\item [-] the initial wavefunction has a Gaussian shape centered at $x= \frac 12 a$:
\be
\psi_{G,0} (x) = C \left [ e^{-(x-a/2)^2/\sigma^2 } - e^{-a^2/4\sigma^2 } \right ] \chi_{(0,a)} (x)
\ee
where $\sigma >0$ and 
\be
C = \sqrt {2} \left [ \sqrt {2\pi} \sigma \mbox {erf} \left ( \frac {a}{\sqrt {2} \sigma }\right ) - 4 \sqrt {\pi} \sigma e^{-a^2/4\sigma^2} \mbox {erf} \left ( \frac {a}{ {2} \sigma }\right ) + 2 a  e^{-a^2/2\sigma^2} \right ]^{-1/2} 
\ee
is a normalization parameter; for argument's sake let $\sigma = 0.2$, $a=1$ and thus  $C=2.0028$.

\item [-] the initial wavefunction is a piece-wise constant function:
\be 
\psi_{S,0} (x) =  \sqrt {\frac {2}{a}} \chi_{(a/4,3a/4)} (x)\, .
\ee

\end {itemize}

We next compute the time evolution $\psi_{i,t}$, $i=A,G,S$, for different values of $\eta$ by means of the spectral splitting method and go on to plot the absolute value of the survival amplitude
\be
{\mathcal A}_i (t) = \langle \psi_{i,t}, \psi_{i,0} \rangle \, . 
\ee
In particular, we see that, after a transient time, the term $|{\mathcal A}_i (t) |$ has an exponential    behavior of the kind $b_i e^{-a_i t}$, with $a_i=a_i(\eta ) >0$ real-valued, and we estimate the values of these parameters in Table \ref {Tavola5}.

\subsection {Comparison of numerical experiments}
It appears from Figure \ref {Fig_8} that for small $t$ an oscillating behavior, observed since Winter's article \cite {Wi}, of $|{\mathcal A}_i (t)|$ can occur; this oscillating behavior is clearly explained in the linear case where $\eta =0$ by the fact that in Theorem \ref {Teo2} several resonances can contribute to the wave function $\psi_t$ for small times and thus a typical interference effect may occur. \ After a transient time only the contribution due to the narrow resonances survives and thus the oscillatory behavior disappears and the contribution of the resonances to $|{\mathcal A}_i (t)|$ has the simple form $b_i e^{-a_i t}$. \ Similar oscillatory behavior is also observed for nonlinear models. \ Furthermore, for the nonlinear models we see that the pictures corresponding to $\eta>0$ and $\eta <0$ are quite different in all three cases $i=A,G,S$. \ When $\eta >0$ the dispersion effect seems to be stronger than in the linear case where $\eta =0$, while the dispersion effect seems to gradually vanishes when $\eta$ is negative and, in absolute value, quite large. \ In particular, when $\eta$ reaches the threshold value $\tilde \eta(\alpha )=-6.59$ for $\alpha =10$ the decay effect disappears. \ We also point out that the coefficients $a_i (\eta )$ depend on $\eta =0, \pm 3, \pm 9$ but are essentially the same for any $i=A,G,S$.
\begin{figure}
\begin{center}
\includegraphics[height=4cm,width=4cm]{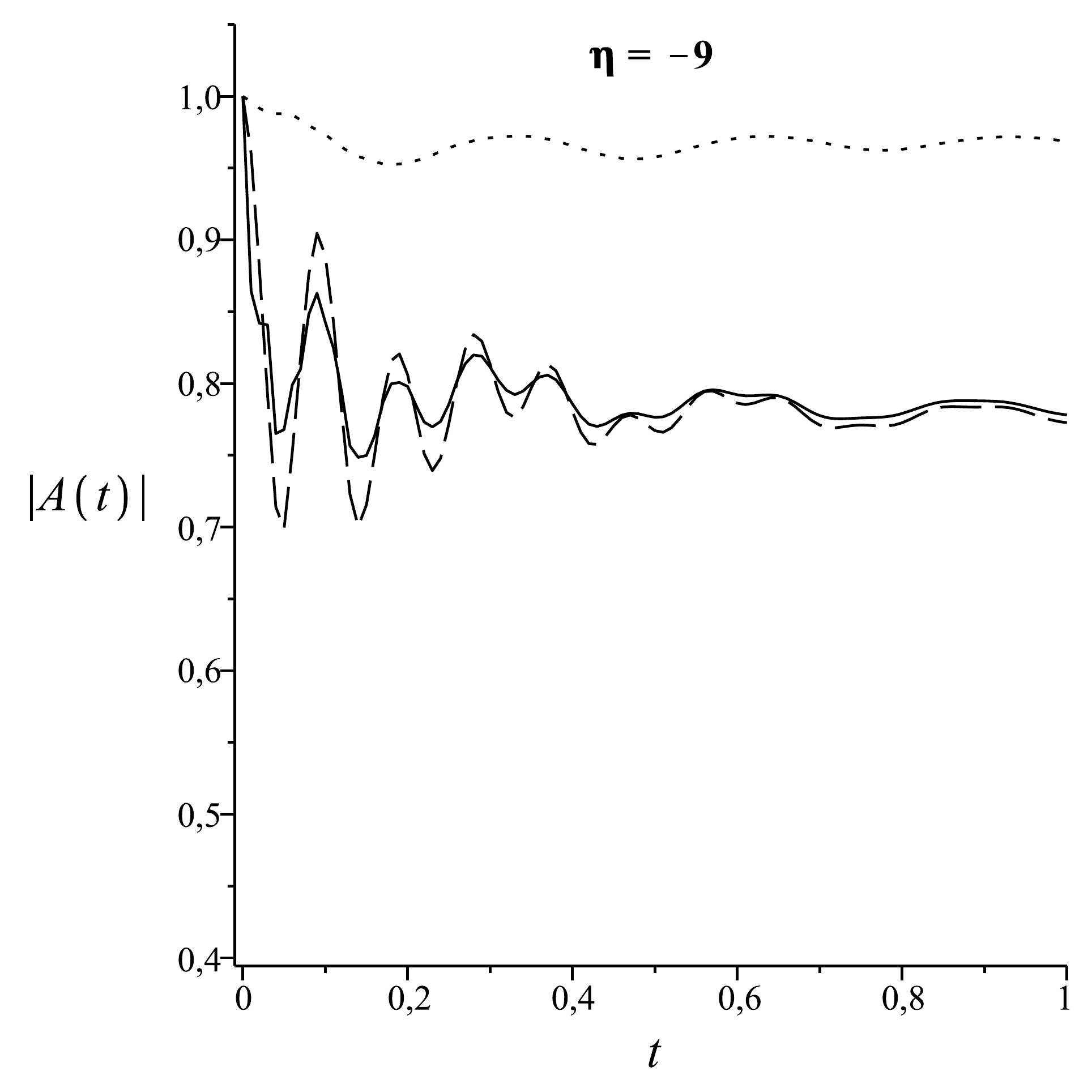}
\includegraphics[height=4cm,width=4cm]{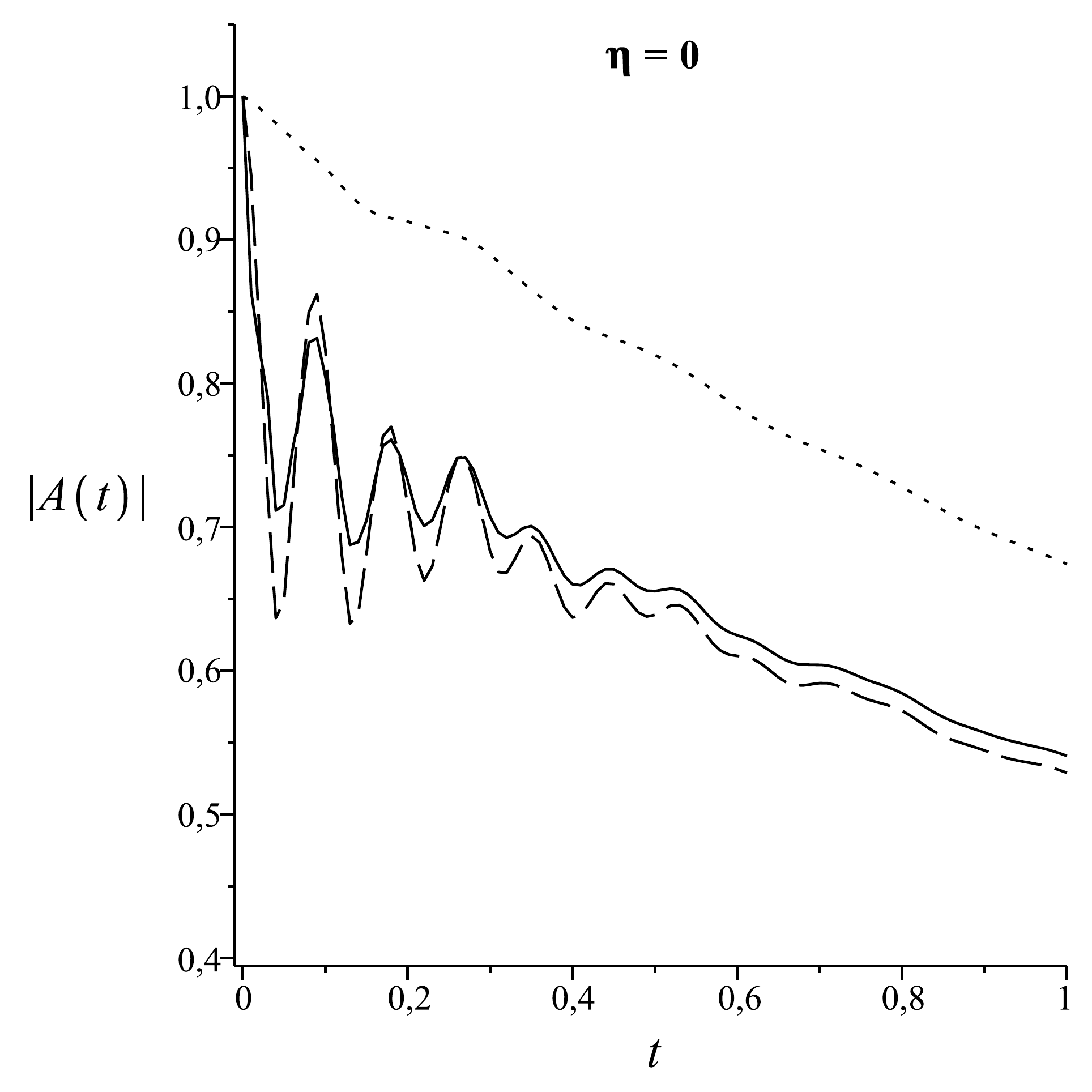}
\includegraphics[height=4cm,width=4cm]{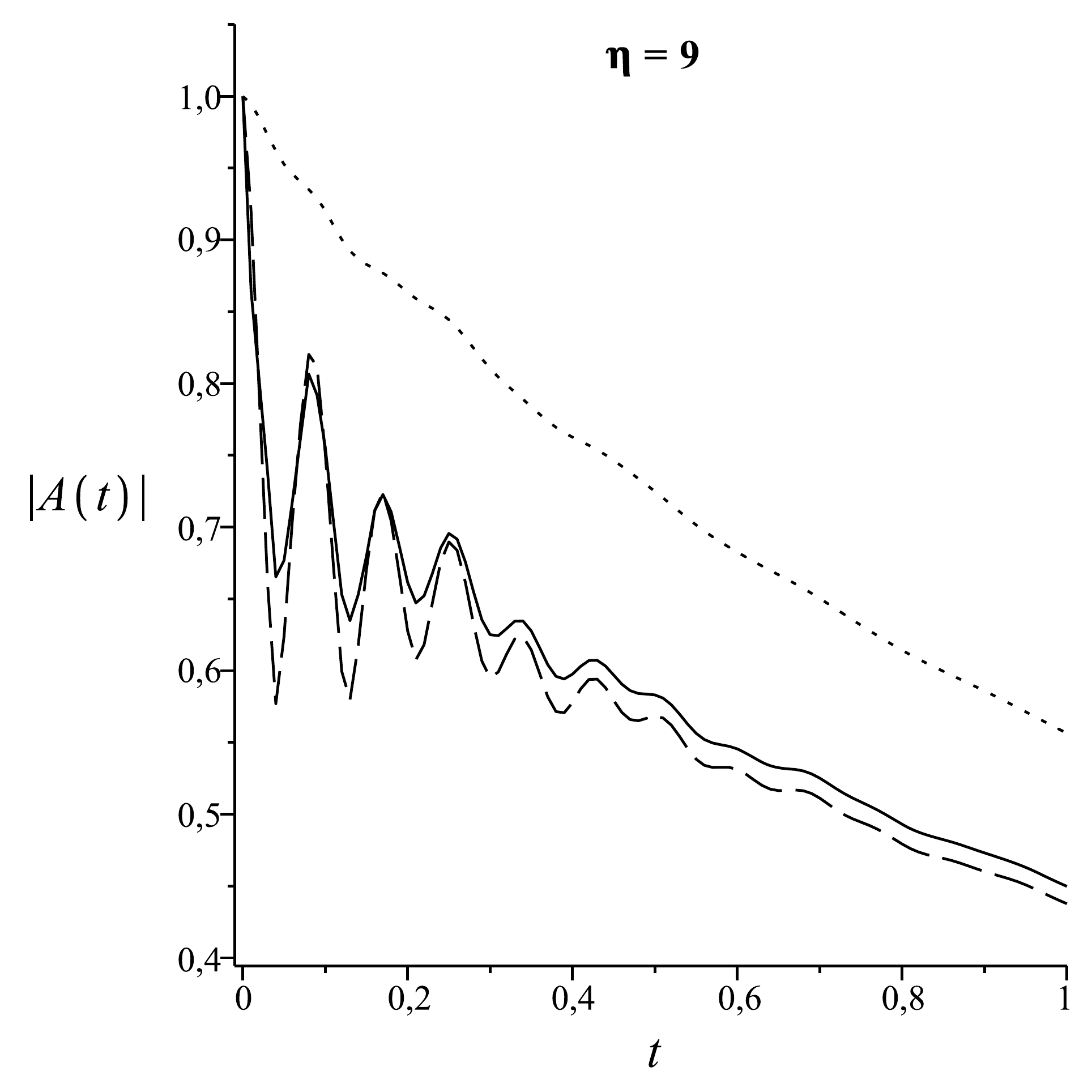}
\caption{Here we plot the absolutely value of the survival amplitude for different values of $\eta$ (i.e. for $\eta =0$, linear case, and for $\eta =\pm 9$). \ The initial wavefunction is $\psi_{A,0}$ (dot line), $\psi_{G,0}$ (broken line) and $\psi_{S,0}$ (full line). \ In particular, $a=1$ and $\alpha =+10$ are fixed.}
\label {Fig_8}
\end{center}
\end{figure}

\begin{table}
\begin{center}
\begin{tabular}{|l||c|c|c|c|c|} \hline
 & $\eta =-9$  &  $\eta =-3 $ & $\eta =0$   & $ \eta =+3 $   &  $\eta =+9 $ \\  \hline
 $a_A$ & $0.0054$ & $0.2877$ & $0.3879$ & $0.4621$  & $0.5770$  \\ \hline
 $b_A$ & $0.9704$ & $0.9927$ & $0.9909$ & $0.9855$  & $0.9723$  \\ \hline
  $a_G$ & $0.0513$ & $0.3403$ & $0.4172$ & $0.4757$  & $0.5659$  \\ \hline
 $b_G$ & $0.8078$ & $0.8022$ & $0.7893$ & $0.7758$  & $0.7501$  \\ \hline
 $a_S$ & $0.0482$ & $0.3314$ & $0.4092$ & $0.4695$  & $0.5634$  \\ \hline
 $b_S$ & $0.8111$ & $0.8132$ & $0.8033$ & $0.7923$  & $0.7708$  \\ \hline
\end{tabular}
\caption{Fit the absolute value of the survival amplitude function $\Am (t)$ with an exponential function of the kind $b_i e^{-a_i t}$;  the values of $a_i (\eta)$ and $b_i (\eta)$, where $i=A,G,S$, depends on $\eta$ and they are listed above.}
\label{Tavola5}
\end{center}
\end {table}

\subsection {Conclusions}
From the numerical analysis performed we can conclude that an extension to the nonlinear case of the notion of quantum resonance seems admissible and that a possible definition of it as a pole of the scattering coefficient $S(\En )$ deserves to be proposed. \ Unfortunately, it is not possible to verify this definition analytically on a simple model because the meromorphic extension of $S(\En )$ to complex values of $\En$ is not explicitly known in the case where the nonlinear term $\eta |\psi_t |^2 \psi_t$ contains the absolute value. \ Only numerical experiments are, at this moment, possible. \ The numerical analysis also confirm the intuitive fact that in the case of a repulsive nonlinear potential (i.e., with $\eta >0$) the nonlinear Schr\"odinger equation exhibits a more dispersive character and thus the survival amplitude decreases more rapidly over time than in the linear case. \ On the contrary, in the case of an attractive nonlinear potential (i.e., with $\eta <0$), the opposite effect is observed; that is the survival amplitude decreases more slowly until it becomes essentially stable when $\eta$ reaches a critical value. \ It is evident from the graphs that this critical value occurs in conjunction with the presence of stationary solutions. \ Thus, we can conjecture that in the nonlinear case we are in the presence of quantum resonances that similarly affect the decay of the survival amplitude and that the associated energy has a negative imaginary part that tends to increase (in absolute value) as positive $\eta$ increases and that, on the other hand, tends to zero when negative $\eta$ approaches the threshold value; at this limit value the resonance becomes a stationary state.

\appendix 

\section {Proof of the Theorem \ref {Teo1}} \label {App1}The proof is  essentially based on a long and straightforward calculation. \ We collect here the mains steps.

Since
\be
f_n^1 = \frac 1{2 \sqrt {t}} (2an+|x|+|y|) + i\alpha \sqrt {t} /2 \ \mbox { and } \
e_n^1 = e^{i(an+(|x|+|y|)/2)^2/t} \, , 
\ee
then one can check that 
\bee
U_1 (x,y,t)=  -\frac {1}{\sqrt {8\pi}} \sum_{n=0}^\infty \alpha^n   g_n^1 \left ( \frac {it}{2} \right )^{(n-1)/2} e^{-i(f_n^1)^2/2}D_{-n} ((1-i)f_n^1) \, . \label {Formula34}
\eee
And similarly 
\bee
U_2 (x,y,t) &=& \frac {1}{\sqrt {8\pi} } \sum_{n=0}^\infty  \alpha^{n+1} 
(it/2)^{n/2} e^{-i(f_n^2)^2/2}  g_n^2 D_{-n-1} ((1-i)f_n^2) \label {Formula35} \\ 
U_3 (x,y,t) &=& \frac {1}{\sqrt {8\pi} } \sum_{n=0}^\infty  \alpha^{n+1} 
(it/2)^{n/2} e^{-i(f_n^3)^2/2}  g_n^3 D_{-n-1} ((1-i)f_n^3) \label {Formula36} \\
U_4 (x,y,t) &=& - \frac {1}{\sqrt {8\pi} } \sum_{n=0}^\infty  \alpha^{n+1} 
(it/2)^{n/2} e^{-i(f_n^4)^2/2}  g_n^4 D_{-n-1} ((1-i)f_n^4)\, .  \label {Formula37}
\eee

Indeed, for instance, let $U_1 = \frac {i}{4\pi } A_1$ where
\be
A_1 &=& 
 \int_{\R +i0} \frac {1}{k} e^{-ik^2 t} K_1 (x,y,k ) dk = \int_{\R +i0} \frac {1}{k} e^{-ik^2 t}\left [ \Gamma (k) \right ]^{-1}_{1,1} e^{ik (|x|+|y|)} dk \\ 
&=& 2i \sum_{n=0}^\infty  \int_{\R }  e^{-ik^2 t}\left ( \frac {-\alpha e^{2ika}}{2ik-\alpha}\right )^n  e^{ik (|x|+|y|)} dk \\ 
&=& 2i \sum_{n=0}^\infty  (-\alpha)^n A_{1,n} \, ,\ A_{1,n} := \int_{\R }  e^{-ik^2 t+2ikan+ ik (|x|+|y|) } \frac {1}{(2ik-\alpha)^n} dk
\ee

If we set
\be
-k^2 t+2kan+ k (|x|+|y|) 
= - r^2  +(2an+|x|+|y|)^2/4t
\ee
where
\be
r= \sqrt {t} \left (  k - (2an+|x|+|y|)/2t \right ) \, 
\ee
then
\be
A_{1,n} 
= (2i)^{-n} t^{(n-1)/2} e^{i(an+(|x|+|y|)/2)^2/t}   \int_{\R }  e^{-ir^2 } \frac {1}{(r+f_n^1)^n} dr \, . 
\ee
We recall now the following result (see Appendix A \cite {KS}): let $m\in \N$ and $z \in \C$ such that $\Im z >0$, then
\be
\int_{\R} \frac {e^{-ix^2}}{(x+z)^n} dx = -i (-2i)^{(n-1)/2}\sqrt {2\pi} e^{-iz^2/2}D_{-n} ((1-i)z)
\ee
where $D_{-n} (z)$ is the parabolic cylinder function.
Hence,
\be
A_{1,n}
=  - \left (- \frac {t}{2i} \right )^{(n-1)/2} \sqrt {\pi/2} e^{-i(f_n^1)^2/2} g_n^1 D_{-n} ((1-i)f_n^1)  \cdot (-1)^{n+1}  \,  .
\ee
and then (\ref {Formula34}) follows. 

Similarly, we can obtain an expression for $U_j$, $j=2,3,4$, by means of convergent series. \ E.g., let $U_2 = \frac {i}{4\pi }  A_2$ where
\be
 A_2 &=& 
 \int_{\R +i0} \frac {1}{k} e^{-ik^2 t} K_2 (x,y,k ) dk = \int_{\R +i0} \frac {1}{k} e^{-ik^2 t}\left [ \Gamma (k) \right ]^{-1}_{1,2} e^{ik (|x|+|y-a|)} dk \\ 
&=&  2i \alpha \int_{\R +i0}\frac {1}{2ik-\alpha}  \frac {1}{1+ \frac {\alpha e^{2ika}}{2ik-\alpha}}  e^{-ik^2 t}  e^{ik (|x|+|y-a|+a)} dk 
\\ 
&=& 2i \alpha \sum_{n=0}^\infty  \int_{\R }  \frac {1}{2ik-\alpha} \left ( \frac {-\alpha e^{2ika}}{2ik-\alpha}\right )^n   e^{-ik^2 t} e^{ik (|x|+|y-a|+a)} dk \\ 
&=& 2i\alpha \sum_{n=0}^\infty  (-\alpha)^n A_{2,n} \, ,
\ee
where
\be
A_{2,n} &:=& \int_{\R }  e^{-ik^2 t+2ikan+ ik (|x|+|y-a|+a) } \frac {1}{(2ik-\alpha)^{n+1}} dk
\ee

If we set
\be
-k^2 t+2ka n + k (|x|+|y-a|+a) 
= - r^2  +(2an+|x|+|y-a|+a)^2/4t
\ee
where
\be
r= \sqrt {t} \left (  k - (2an+|x|+|y-a|+a)/2t \right ) \, .
\ee
then
\be
A_{2,n} 
&=& (2i)^{-n-1} t^{n/2}  e^{i(an+(|x|+|y-a|+a)/2)^2/t}   \int_{\R }  e^{-ir^2 } \frac {1}{(r+f_n^2)^{n+1}} dr \\
& =& (2i)^{-n-1} t^{n/2}  e^{i(an+(|x|+|y-a|+a)/2)^2/t}   \left [ -i (-2i)^{n/2}\sqrt {2\pi} e^{-i(f_n^2)^2/2}D_{-n-1} ((1-i)f_n^2) \right ]  \\
&=&\sqrt {\pi/2}  (-1)^{n+1} (it/2)^{n/2} e^{-i(f_n^2)^2/2} g_n^2 D_{-n-1} ((1-i)f_n^2) \, .
\ee
In conclusion (\ref {Formula35}) follows.

Similarly, we obtain the expression for $U_3$ and $U_4$. \ Eventually, Theorem \ref {Teo1} follows.

\end {document}